\documentclass[a4paper,fleqn]{cas-sc}

\usepackage[numbers, sort&compress]{natbib}

\usepackage{amsmath}   
\usepackage{amsfonts}  
\usepackage{amssymb}   
\usepackage{amsthm}    
\usepackage{enumitem}  

\theoremstyle{plain} 
\newtheorem{proposition}{Proposition}[section] 
\theoremstyle{definition} 
\newtheorem{remark}{Remark}[section] 
\theoremstyle{definition} 

\DeclareMathOperator{\Span}{span}

\def\tsc#1{\csdef{#1}{\textsc{\lowercase{#1}}\xspace}}
\tsc{WGM}
\tsc{QE}

\begin{document}
\let\WriteBookmarks\relax
\def\floatpagepagefraction{1}
\def\textpagefraction{.001}

\shorttitle{Snapping and Boundary Enforcement in Model Reduction}

\shortauthors{Farhat et al.}  

\title [mode = title]{The Snapping Phenomenon and Boundary Enforcement in Model Reduction: Mechanisms and Mitigation}  

\author[1]{Charbel Farhat}\cormark[1]
\ead{cfarhat@stanford.edu}
\credit{Conceptualization, Methodology, Formal Analysis, Writing - original draft, Supervision, Funding acquisition}

\author[1]{Yihong Zhu}
\ead{zyh03@stanford.edu}
\credit{Software, Validation, Investigation, Visualization}

\author[1]{Faisal As'ad}
\ead{faisal3@stanford.edu}
\credit{Methodology, Validation, Investigation, Visualization}

\cortext[1]{Corresponding author}

\affiliation[1]{organization={Stanford University},
  addressline={William F. Durand Building, 496 Lomita Mall}, 
  city={Stanford},
  postcode={94305-4035}, 
  state={CA},
  country={USA}}
 
\author[2]{Tommaso Taddei}
\ead{tommaso.taddei@uniroma1.it}
\credit{Conceptualization, Formal Analysis}
\affiliation[2]{organization={Department of Mathematics, Sapienza University of Rome},
  addressline={Piazzale Aldo Moro, 5},
  city={Rome},
  postcode={00185},
  country={Italy}}

\begin{abstract}
Projection-based model order reduction is a vital paradigm for mitigating the prohibitive computational overhead of 
	high-fidelity, parametric engineering simulations. However, standard reduction workflows remain vulnerable to a severe failure
	mode, termed here the ``snapping phenomenon,'' while simultaneously struggling to enforce parameter-dependent boundary 
	conditions without resorting to expensive online constraints or penalty schemes. Snapping manifests as a pathologically rigid,
	non-physical latching of the reduced-order model onto a specific training snapshot during online inference, severely inhibiting
	smooth parametric regression or interpolation. This work introduces the first comprehensive tripartite taxonomy of the 
	phenomenon, classifying it into intrinsic, extrinsic, and numerical snapping based on its root causes. Focusing structurally 
	on isolating intrinsic (Type~1) and numerical (Type~3) snapping mechanisms, we introduce a unified mathematical framework 
	utilizing a parameter-dependent shift vector that addresses both vulnerabilities simultaneously. Unlike a posteriori 
	constraint techniques, such as Lagrange multipliers that increase online computational complexity or penalty methods that 
	require cumbersome parameter tuning, our framework embeds an explicit spatial transport mechanism directly into the state 
	ansatz. This guarantees exact, a priori satisfaction of parameter-dependent boundary conditions across unseen queries while 
	restoring continuous parametric sensitivity and eliminating snapshot latching in both standard affine and regression-augmented
	architectures. We first establish the rigorous mathematical theory of this framework within the linear regime, proving how it 
	breaks test-space orthogonality blind spots. We then extend this account to the nonlinear regime, demonstrating how numerical 
	snapping in external computational fluid dynamics is driven by the interplay of parameter-dependent boundary conditions, mesh 
	grading, and discrete equation scaling. We critique historical formulation-toggling, showing that it operates as an unreliable
	empirical heuristic that compromises near-field fidelity targets. To overcome these limitations, we propose a field-theoretic 
	realization utilizing a spatially-varying shift vector governed by a scaled Laplace framework. The proposed framework is 
	validated across structural analysis and aerodynamic benchmarks, including parametric viscous flows around a NACA airfoil and 
	a three-dimensional NASA Common Research Model configuration.
\end{abstract}

\begin{highlights}
  \item \textbf{Tripartite Taxonomy of Snapping:} Establishes a novel tripartite classification of the snapshot snapping phenomenon, distinguishing Intrinsic (Type 1), Extrinsic (Type 2), and Numerical (Type 3) failure modes in 
	  projection-based reduced-order models.
  \item \textbf{State-Shifting Framework for Boundary Conditions:} Introduces a unified state-shifting formulation that embeds parametric spatial transport, breaking test-space orthogonality blind spots and enforcing 
	  parameter-dependent boundary conditions a priori.
  \item \textbf{CFD Realization \& Industrial Scalability:} Develops a scaled Laplace extension scheme for non-descriptor CFD models, demonstrating accurate, constraint-free parameter extrapolation on industrial meshes up to 58 
	  million elements.
\end{highlights}


\maketitle

\section{Introduction}\label{sec:intro}
High-fidelity High-Dimensional Models (HDMs), such as those based on large-scale Finite Volume (FV) or Finite Element (FE) discretizations of Partial Differential Equations (PDEs), provide unprecedented accuracy in simulating complex physical phenomena. 
However, comprehensive engineering workflows require thorough \textit{parametric analyses}, wherein system responses must be evaluated across vast, multidimensional parameter spaces involving variations in geometries, boundary conditions, or material properties. The extreme computational overhead of repeated HDM evaluations renders such parametric studies prohibitive for many-query applications, 
including real-time control, optimization, and uncertainty quantification. To alleviate this burden, Projection-based Model Order Reduction (PMOR) has emerged as a crucial paradigm. By projecting high-dimensional state spaces onto parameter-independent, 
low-dimensional subspaces, PMOR constructs efficient Projection-based Reduced-Order Models (PROMs) capable of rapid, multi-query evaluation.

In the development of PMOR techniques, two primary modeling contexts dictate the reduction strategy: parametric linear problems and 
parametric nonlinear problems. Within these contexts, the underlying HDMs are typically cast into one of two mathematical formulations:
the descriptor form or the non-descriptor form~\cite{amsallem2014stability}. In the linear regime, intrinsic snapping arises purely 
from the underlying projection mechanics and parametric structural changes of the PROM, occurring primarily independent of 
discretization artifacts. Conversely, in the nonlinear regime, the choice between these two structural formulations becomes intimately
bound to spatial error distribution and numerical scaling~\cite{amsallem2014stability}.

The descriptor form retains algebraic constraints or the generalized mass matrix explicitly alongside the differential equations, 
preserving the original, unmanipulated structure of the governing physical laws. While analytically straightforward because it avoids 
the algebraic manipulation or inversion of a generalized mass matrix, the descriptor form is frequently avoided in practical HDM 
simulations. In many general applications, this avoidance stems from severe numerical and algorithmic complexities; the singular or 
poorly conditioned nature of the leading coefficient matrix forces the use of computationally expensive implicit differential-algebraic
equation solvers~\cite{wanner1996solving} and generates ill-conditioned linear systems that hinder solver convergence. In external 
Computational Fluid Dynamics (CFD) flow applications, however, the motivation for avoiding the descriptor form is driven by a 
fundamentally different numerical requirement related to spatial error distribution.

Conversely, the non-descriptor form isolates the time derivatives to express the system purely in terms of explicit state residuals. 
When the generalized mass matrix $\mathbf{M}$ is non-singular and easily invertible (such as a diagonal or well-conditioned lumped mass
matrix), the non-descriptor form is obtained directly by premultiplying the entire system by $\mathbf{M}^{-1}$, thereby isolating the 
leading time-derivative term. The non-descriptor form serves as the standard convention for many engineering simulations and becomes 
uniquely advantageous in CFD when the computational domain contains an obstacle. This geometric configuration establishes a clear 
distinction between a near-field and a far-field solution. In these frameworks, the far-field region typically features large mesh 
cells (or elements) because the solution there is nearly uniform, whereas the near-field region requires highly refined, small cells 
to resolve steep localized gradients. Because the non-descriptor form naturally scales the governing residual by the inverse of the 
cell volume, it prioritizes error minimization in the highly refined regions, thereby guaranteeing maximum near-field accuracy for a 
given convergence tolerance.

With the structural distinctions between the descriptor and non-descriptor forms established, we turn our attention to two critical, 
interconnected vulnerabilities in standard PMOR workflows that arise independently of the choice of mathematical formulation. First, 
standard reduced-order trial spaces do not, in general, satisfy unseen, parameter-dependent boundary conditions exactly. This 
challenge has been addressed through approaches such as particular-solution constructions and boundary-compatible reduced 
spaces~\cite{gunzburger2007reduced}, while alternative strategies include penalty formulations that require delicate parameter 
tuning~\cite{cosimo2016general} and \textit{a posteriori} constraint-enforcement schemes, such as Lagrange multipliers, that can 
increase online system dimensions and computational costs~\cite{kalashnikova2014penalty}. Second, residual-minimizing PMOR pipelines 
remain vulnerable to a severe failure mode that, apart from two isolated discussions and a preliminary remedy suggested or implied 
in~\cite{washabaugh2014enforcing,washabaugh2016faster}, has received little systematic investigation in the model reduction literature.

With the structural distinctions between the descriptor and non-descriptor forms established, we turn our attention to two critical, 
interconnected vulnerabilities in standard PMOR workflows that operate, \textit{a priori}, independently of the choice of mathematical
formulation. First, standard reduced-order trial spaces struggle to satisfy unseen, parameter-dependent boundary conditions without 
resorting to penalty formulations that require delicate parameter tuning~\cite{cosimo2016general} or \textit{a posteriori} constraint 
schemes (such as Lagrange multipliers) that inherently inflate online system dimensions and solve times~\cite{kalashnikova2014penalty}.
Second, residual-minimizing PMOR pipelines remain highly vulnerable to a severe failure mode, which except for two isolated 
discussions and a preliminary remedy suggested or implied in~\cite{washabaugh2014enforcing} and~\cite{washabaugh2016faster}, has been 
omitted from systematic investigation in the model reduction literature. We term this behavior the ``snapping phenomenon'' and argue 
that its rarity in academic discussion belies the severe challenge it poses to PMOR reliability. This vulnerability is particularly 
pronounced in PMOR frameworks that compute generalized coordinates by minimizing the governing residual~\cite{washabaugh2014enforcing},
such as the Least-Squares Petrov-Galerkin (LSPG) method~\cite{carlberg2011efficient, carlberg2013gnat}. We define snapping as the 
rigid, non-physical latching of a surrogate model -- and more specifically, a PROM in the context of this work -- onto a specific 
training snapshot during online inference, which inhibits the model from producing a smooth, continuous regression or interpolation 
across the parameter space. To lay the groundwork for systematically isolating and mitigating this behavior, this work introduces the 
first comprehensive tripartite taxonomy of the snapping phenomenon, categorizing it based on its fundamental root cause:

\begin{itemize} 
	\item \textit{Intrinsic snapping:} Driven directly by the inherent mathematical structure and projection mechanics of the 
		PROM itself.  
	\item \textit{Extrinsic snapping:} Driven by specific modeling choices, hyperreduction methods, or localized reduction 
		constructs introduced during the offline PMOR workflow.
	\item \textit{Numerical snapping:} Driven by discrete computational artifacts, solver tolerances, mesh grading, or spatial 
		scaling choices within the underlying HDM.
\end{itemize} 

Fundamentally, Type~1 (intrinsic) snapping manifests as a byproduct of the projection mechanics themselves, occurring 
entirely independent of spatial discretization artifacts or numerical scaling choices. While this failure mode can be most cleanly 
isolated and mathematically characterized within the linear regime, it represents a structural vulnerability that emerges whenever the
choice of the left test space inadvertently filters out parameter-induced variations of the high-dimensional operators. When a 
Petrov-Galerkin or Galerkin projection scheme maps the high-dimensional state space onto the reduced-order trial subspace, an 
ill-conditioned operator can experience a structural loss of parametric sensitivity. Consequently, during online parameter sweeps, 
the reduced-order model becomes completely blind to variations along specific geometric or material manifolds, causing the 
generalized coordinates to stagnate. This generalized coordinate stagnation represents a pathological locking state in which the 
online solver continuously returns an identical generalized coordinate vector $\mathbf{y}(\boldsymbol{\mu})$ for distinct values of 
the parameter vector $\boldsymbol{\mu}$, because the parameter-induced residual perturbations have been completely annihilated by the 
left test basis. This insensitivity ultimately forces the reconstructed solution to rigidly lock onto an individual offline training 
state, creating a discontinuous, step-like response across that neighborhood of the parameter domain rather than a smooth, physical 
regression.

While Type~1 (intrinsic) snapping stems from underlying projection mechanics -- and Type~3 (numerical) snapping, 
discussed below, arises from grid-resolution scaling -- Type~2 snapping is driven directly by the localized reduction constructs and 
hyperreduction choices introduced during the offline workflow. In highly nonlinear regimes, evaluating the full high-dimensional 
residual operator becomes computationally prohibitive. This forces practitioners to employ hyperreduction methods -- such as the 
Discrete Empirical Interpolation Method (DEIM)~\cite{chaturantabut2010nonlinear} or the Energy Conserving Sampling and Weighting (ECSW)
method~\cite{farhat2014dimensional, farhat2015structure} -- to evaluate residuals exclusively over a sparse, localized sample mesh. 
Type~2 snapping manifests when these sparse sampling patterns or localized latent-space approximations inadvertently introduce spatial
blind spots or artificial residual discontinuities as parameters vary. Because the restricted sample mesh cannot smoothly capture 
moving physical features or localized gradients across the parameter space, the online minimization framework can suffer a severe 
loss of parametric sensitivity. Consequently, the generalized coordinates pathologically collapse and lock onto the nearest training 
snapshot.

In contrast, Type~3 (numerical) snapping frequently arises as a byproduct of discretization and formulation interactions 
within the high-fidelity snapshot generator. In aerodynamic CFD applications, this failure mode typically manifests when the parameter
space includes variables that influence far-field boundary conditions explicitly. In the descriptor form, the system remains relatively
stable against snapping because the explicit parametric dependence at the far-field boundary matches the large physical volumes of 
those cells. However, when transitioning to the non-descriptor form, the governing residual is premultiplied by the inverse of the 
generalized mass matrix, which scales inversely with the cell volumes ($1/\Omega_{\mathrm{cell}}$). Because a conventionally graded 
mesh features highly refined, ultra-small cells in the near-field boundary layer, this inversion introduces a massive amplification 
factor to the near-field residuals. Consequently, the implicit parametric dependence of the near-field is artificially blown up, 
severely dominating the explicit parametric tracking of the far-field. During online evaluation, this severe residual imbalance 
inhibits smooth regression, forcing PROMs -- particularly those architected specifically around residual minimization -- to rigidly 
snap to the nearest sampled training point.

Conversely, an analogous yet inverted mechanism occurs when the parameter space governs an explicit change in the wall boundary, such 
as a geometric shape parameter. In this scenario, the roles of the descriptor and non-descriptor formulations are entirely reversed. 
Because the explicit parametric dependence is now localized in the highly refined near-field, the unweighted residuals of the 
descriptor form allow the massive, implicitly dependent far-field cell volumes to dominate and mask the geometric variations. Here, 
the descriptor form triggers snapping, whereas the non-descriptor form acts as a stabilizing agent; by scaling the residual by 
$1/\Omega_{\mathrm{cell}}$, it naturally amplifies the near-field variations, properly aligning the solver's numerical weight with the
explicit geometric parameterization. Notably, across both configurations, this pathologically rigid behavior remains a pure numerical
artifact that does not reflect the structural configuration of the PROM, nor is it a byproduct of the reduction constructs introduced 
during PMOR.

The scope of the present work is strictly limited to characterizing and mitigating Type~1 (intrinsic) and Type~3 
(numerical) snapping; a formal investigation into Type~2 (extrinsic) snapping is deferred to a future study. 
Methodologically, this work focuses on parametric, steady external flow CFD problems where the parameter space dictates 
parameter-dependent far-field boundary conditions, wall boundaries, or a combination thereof~\cite{de2026nonlinear}. However, the 
fundamental insights and mitigation strategies developed herein are applicable, in principle, to any steady or unsteady computational 
framework developed for the numerical solution of PDEs exhibiting analogous discretization and parameterization characteristics.

In seeking to address these specific failure modes, a critique of the existing literature reveals that, based on the results presented
in~\cite{washabaugh2014enforcing, washabaugh2016faster}, the historical approach developed within our research group has been to toggle
between the descriptor and non-descriptor formulations depending on which boundary conditions are affected by the parameter space. 
However, we identify a fundamental and irreconcilable limitation to this conventional paradigm. While toggling remains viable as an 
empirical heuristic for isolated parameter sets, it completely breaks down when the parameter space simultaneously governs both 
far-field conditions and near-field wall characteristics (such as surface geometry modifications or wall boundary conditions); because
one regime strictly demands the descriptor form while the other requires the non-descriptor form, the dual-parameter problem leaves the
practitioner structurally deadlocked. Furthermore, even in isolated settings, this toggling strategy introduces severe compromises: 
forcing a descriptor formulation inherently degrades essential near-field fidelity targets, and our newly proposed mechanistic account
demonstrates that it can still fail to prevent snapping under routine operational constraints. Here, we define such routine constraints
as standard production environments where the high-dimensional mesh grading is engineered solely to satisfy localized simulation 
accuracy targets without \textit{a priori} anticipation of snapping artifacts, spanning parameter dimensions and ranges that vary 
arbitrarily in scale.

To overcome both the operational limits of formulation-toggling and the challenge of parametric boundary enforcement, we introduce a 
unified mathematical framework utilizing a parameter-dependent shift vector. Unlike \textit{a posteriori} constraint methodologies 
that enforce boundary conditions at the expense of online optimization complexity, our framework embeds dynamic parameter transport 
directly into the state representation prior to residual minimization. This dual-action construction restores continuous parametric 
sensitivity while enforcing inhomogeneous boundary conditions \textit{a priori} across both operational regimes:
\begin{itemize} 
	\item \textit{The linear regime (intrinsic snapping):} We present a rigorous mathematical theory explaining how intrinsic 
		structural features drive the snapping phenomenon. We then introduce an algebraic realization of the 
		state-reconstruction shift vector to restore the model's sensitivity to parameter changes, overcoming the 
		pathologically rigid generalized coordinate stagnation that characterizes standard projections.  
	\item \textit{The nonlinear regime (numerical snapping):} We provide a mechanistic account of how parameter-dependent boundary
		conditions, spatial mesh grading, and discrete equation scaling combine to drive numerical snapping in FV external 
		flow CFD formulations utilizing dual cells. We then generalize our state-shifting framework into a field-theoretic 
		mitigation strategy governed by a scaled Laplace framework to smoothly diffuse boundary transformations into the 
		interior volume mesh, satisfying parameter-dependent boundary conditions \textit{a priori} and neutralizing Type~3 
		numerical snapping under the non-descriptor formulation.
\end{itemize}

The remainder of this paper is organized as follows. Section~\ref{sec:framework} introduces the general mathematical formulation of 
our unified state-shifting framework, outlining how a parameter-dependent shift vector enforces boundary updates \textit{a priori} and
restores continuous parametric sensitivity to the state representation. Section~\ref{sec:linear} then applies this framework to 
parametric linear problems, detailing the theoretical foundations of intrinsic snapping and establishing our formal mitigation 
approach. Section~\ref{sec:nonlinear} shifts focus to parametric nonlinear problems, establishing the physical mechanism of numerical 
snapping driven by mesh-resolution scaling in CFD, evaluating the performance and accuracy limitations of historical descriptor 
configurations, and introducing our specific realization utilizing a spatially-varying shift vector governed by a Laplace framework. 
Section~\ref{sec:applications} validates these theoretical and mechanistic frameworks across a series of rigorous numerical benchmarks,
including a linear static structural analysis, a parametric two-dimensional viscous airfoil, and a hyperreduced three-dimensional NASA
Common Research Model (CRM) configuration. Concluding the paper, Section~\ref{sec:conclusions} provides brief remarks and outlines 
directions for future research. 

\section{The unified state-shifting framework}\label{sec:framework}
To overcome the localized loss of parametric sensitivity that characterizes standard PMOR pipelines and eliminate the need for 
computationally heavy \textit{a posteriori} boundary constraint schemes, we here introduce a generalized mathematical abstraction 
termed the \textit{unified state-shifting framework}. Standard reduction methodologies depend on a traditional affine representation 
of the high-dimensional state space, where parameter-dependent variations are captured exclusively through a low-dimensional coordinate
vector moving within a static subspace. When the projection mechanics filter out operator variations, or when discrete grid resolution
disparities introduce massive numerical imbalances, this static representation collapses into the pathologically rigid snapshot 
latching detailed in this work.

Rather than treating these distinct structural and numerical failure modes with decoupled, localized algorithmic patches, our framework
alters the foundational trial space representation itself. We propose an augmented affine ansatz that embeds an explicit, 
parameter-dependent spatial tracking mechanism directly into the high-dimensional state representation prior to any projection or 
residual minimization step.

Let $\mathbf{x}(\boldsymbol{\mu}) \in \mathbb{R}^N$ represent the high-dimensional physical state vector governed by a general 
parametric system, where the state dimension $N$ is assumed to be prohibitively large, and where $\boldsymbol{\mu} \in \mathcal{D} 
\subset \mathbb{R}^{N_{\mu}}$ denotes the vector of independent parameters operating within an $N_{\mu}$-dimensional parameter domain.
We redefine the trial space approximation by introducing a parameter-dependent, spatially-varying shift vector $\bar{\mathbf{x}}
(\boldsymbol{\mu}) \in \mathbb{R}^N$:
\begin{equation}
	\mathbf{x}(\boldsymbol{\mu}) \approx \tilde{\mathbf{x}}(\boldsymbol{\mu}) = \mathbf{V} \mathbf{y}(\boldsymbol{\mu}) + 
	\mathbf{x}_{\mathrm{ref}} + \bar{\mathbf{x}}(\boldsymbol{\mu}) 
	\label{eq:affine_approx_shifted}
\end{equation}
where $\mathbf{V} \in \mathbb{R}^{N \times n}$ is the parameter-independent right reduced-order basis (ROB), $n \ll N$, 
$\mathbf{x}_{\mathrm{ref}} \in \mathbb{R}^N$ is the static reference state or snapshot mean, and $\mathbf{y}(\boldsymbol{\mu}) \in 
\mathbb{R}^n$ is the vector of generalized coordinates computed online. 

\begin{remark}[Separation of static reference state and dynamic shift vector]
	While $\mathbf{x}_{\mathrm{ref}}$ and $\bar{\mathbf{x}}(\boldsymbol{\mu})$ could abstractly be combined into a single 
	consolidated vector $\bar{\mathbf{x}}_{\mathrm{total}}(\boldsymbol{\mu}) = \mathbf{x}_{\mathrm{ref}} + 
	\bar{\mathbf{x}}(\boldsymbol{\mu})$, maintaining their explicit separation serves critical conceptual and computational 
	purposes: 
	\begin{enumerate} 
		\item \textit{Isolation of dynamic parametric transport:} Keeping $\bar{\mathbf{x}}(\boldsymbol{\mu})$ distinct 
			cleanly isolates dynamic parametric transport from the baseline physical state. Importantly, maintaining this 
			parameter-dependent shift explicitly in the general formulation provides the required baseline to satisfy 
			snapshot reproduction at training points (up to subspace truncation error) and to rigorously analyze snapshot 
			reconstruction errors and generalized coordinate stagnation when snapping occurs.
		\item \textit{Offline/online computational efficiency:} In online implementations, $\mathbf{x}_{\mathrm{ref}}$ 
			generates static, parameter-independent projected vectors that are pre-assembled once offline. Separating 
			$\mathbf{x}_{\mathrm{ref}}$ explicitly exposes this decoupling, ensuring no redundant high-dimensional 
			operations occur during online query evaluation.  
		\item \textit{Conceptual separation of physical states and parametric shifts:} In spatial PDE settings, 
			$\mathbf{x}_{\mathrm{ref}}$ represents a reference physical state vector (e.g., a fluid or structural state), 
			whereas $\bar{\mathbf{x}}(\boldsymbol{\mu})$ represents a parameter-dependent spatial shift. Retaining them 
			as distinct entities prevents conceptual conflation between physical state variables and parameter-driven 
			spatial transformations.
	\end{enumerate} 
\end{remark}

The mathematical purpose of the shift vector $\bar{\mathbf{x}}(\boldsymbol{\mu})$ is to act as a continuous algebraic transport 
mechanism across the parameter domain $\mathcal{D}$. Structurally, the shift vector is constructed to satisfy two foundational 
requirements:
\begin{enumerate} 
	\item \textit{Snapshot reproduction at training points:} At any offline training snapshot location 
		$\boldsymbol{\mu}_{i} \in \{\boldsymbol{\mu}_{j}\}_{j=1}^{n_{s}}$, the shift vector $\bar{\mathbf{x}}
		(\boldsymbol{\mu})$ is constrained to ensure that the reconstruction ansatz~\eqref{eq:affine_approx_shifted} 
		preserves the integrity of the sampled states (and guarantees exact data recovery under an untruncated ROB).
		The algebraic manifestation of this condition is highly dependent on the operational regime: 
		\begin{itemize} 
			\item In the linear operator regime, this requirement is satisfied by forcing the shift to vanish, yielding 
				$\bar{\mathbf{x}}(\boldsymbol{\mu}_{i}) = \mathbf{0}$, thereby leaving the baseline snapshot 
				representation unperturbed.  
			\item In the nonlinear regime, the shift vector at a training point recovers localized state discrepancies or 
				boundary shifts, mapping directly to the \textit{latent or closure error} evaluated at the snapshot 
				coordinates.  
		\end{itemize} 
		In both contexts, this constraint guarantees that the underlying reduced subspace accurately anchors the continuous 
		parameter tracking mechanism without introducing unphysical artifacts or modeling bias.
	\item \textit{Restoration of explicit parameter tracking and exact boundary enforcement:} For an unseen online query 
		parameter $\boldsymbol{\mu}^\star \notin \{\boldsymbol{\mu}_j\}_{j=1}^{n_s}$, the shift vector 
		$\bar{\mathbf{x}}(\boldsymbol{\mu})$ dynamically deforms the high-dimensional state representation, mapping changes in
		boundary conditions, operating parameters, or localized geometric features directly onto the state vector. Crucially, 
		$\bar{\mathbf{x}}(\boldsymbol{\mu})$ is explicitly engineered to satisfy parameter-dependent boundary constraints 
		(such as far-field free-stream updates or near-field wall characteristics) \textit{a priori}. By absorbing these 
		boundary variations directly into the shifted trial space ansatz, the ROB $\mathbf{V}$ is strictly 
		constrained to capture homogeneous fluctuations around a physically valid state, relieving the generalized coordinates
		$\mathbf{y}(\boldsymbol{\mu})$ of the nonlinear burden of tracking boundary updates and guaranteeing exact boundary 
		satisfaction without online constraint penalties.
\end{enumerate}

By substituting the augmented ansatz~\eqref{eq:affine_approx_shifted} into a generic high-dimensional residual operator 
$\mathbf{r}\left(\mathbf{x}(\boldsymbol{\mu}); \boldsymbol{\mu}\right) = \mathbf{0}$, the modified residual system evaluated online reads:
\begin{equation*} 
	\mathbf{r}_{\mathrm{sh}}\left(\mathbf{y}(\boldsymbol{\mu}); \boldsymbol{\mu}\right) = \mathbf{r}\left(\mathbf{V} 
	\mathbf{y}(\boldsymbol{\mu}) + \mathbf{x}_{\mathrm{ref}} + \bar{\mathbf{x}}(\boldsymbol{\mu}); \boldsymbol{\mu}\right) 
\end{equation*}
where the subscript $\mathrm{sh}$ denotes the shifted state representation, defining a composite algebraic operator mapped directly 
onto the low-dimensional generalized coordinates. When a Petrov-Galerkin projection closes this system by enforcing orthogonality 
against a left test space represented by a left ROB $\mathbf{W} \in \mathbb{R}^{N\times n}$, the resulting PROM takes the form:
\begin{equation*} 
	\mathbf{W}^T \mathbf{r}\left(\mathbf{V} \mathbf{y}(\boldsymbol{\mu}) + \mathbf{x}_{\mathrm{ref}} + 
	\bar{\mathbf{x}}(\boldsymbol{\mu}); \boldsymbol{\mu}\right) = \mathbf{0}
\end{equation*}

The fundamental utility of this state-shifting framework lies in how the vector field~$\bar{\mathbf{x}}(\boldsymbol{\mu})$ alters the 
kernel of the governing system while simultaneously guaranteeing exact boundary transport. To be precise, we define the composite 
projected residual operator~$\mathcal{R}_{\boldsymbol{\mu}}: \mathbb{R}^n \rightarrow \mathbb{R}^n$ as:
\begin{equation*} 
	\mathcal{R}_{\boldsymbol{\mu}}(\mathbf{y}) = \mathbf{W}^T\mathbf{r}(\mathbf{V}\mathbf{y} + \mathbf{x}_\mathrm{ref} + 
	\bar{\mathbf{x}}(\boldsymbol{\mu});\boldsymbol{\mu})
\end{equation*}
In standard PROMs where~$\bar{\mathbf{x}}(\boldsymbol{\mu})=\mathbf{0}$, the online solver must find a generalized coordinate 
vector~$\mathbf{y}(\boldsymbol{\mu})$ such that~$\mathcal{R}_{\boldsymbol{\mu}}(\mathbf{y}) = \mathbf{0}$. However, 
parameter-dependent variations of the high-fidelity physics are frequently masked or filtered out because they fall into the massive 
unmodeled subspace (or kernel) associated with the reduced-order evaluation mapping. By introducing the parameter-dependent 
shift~$\bar{\mathbf{x}}(\boldsymbol{\mu})$ directly into the input argument of the residual, we explicitly disrupt these structural 
or numerical orthogonality conditions. This deforms the kernel of~$\mathcal{R}_{\boldsymbol{\mu}}$, eliminating steep boundary-cell 
residual spikes, preventing the pathologically rigid generalized coordinate stagnation that characterizes standard projections, and 
ensuring that the projected system retains a sharp, continuous sensitivity to parameter variations~$\boldsymbol{\mu}$ across the 
entire domain~$\mathcal{D}$.
 
While the algebraic formulation of the shift remains invariant across all applications, its practical realization depends entirely on 
the operational regime under consideration. In the following sections, we demonstrate the versatility of this framework across two 
highly distinct computational settings. In Section~\ref{sec:linear}, we analyze the framework within the linear regime, showing how 
the shift targets and breaks exact algebraic test-space null space conditions. In Section~\ref{sec:nonlinear}, we transition to the 
nonlinear regime, detailing how this general vector field $\bar{\mathbf{x}}(\boldsymbol{\mu})$ is physically realized on complex CFD 
grids by solving a field-theoretic, mesh-resolution-scaled Laplace framework to enforce boundary conditions \textit{a priori} and 
eliminate numerical snapping.

\section{Intrinsic snapping in PMOR for parametric linear problems}
\label{sec:linear}

In this section, we focus specifically on Type~1 (intrinsic) snapping, isolating its theoretical origins and mitigation within the 
linear regime. To detail the precise mathematical mechanisms governing our unified state-shifting framework, we instantiate the 
general parametric operator $\mathbf{r}(\mathbf{x}(\boldsymbol{\mu}); \boldsymbol{\mu}) = \mathbf{0}$ introduced in 
Section~\ref{sec:framework} for linear systems. Specifically, we consider the large-scale, high-dimensional parametric linear system:
\begin{equation}
\mathbf{A}(\boldsymbol{\mu})\mathbf{x}(\boldsymbol{\mu}) = \mathbf{b}(\boldsymbol{\mu}) \label{eq:linear_hdm}
\end{equation}
where the square, non-singular matrix $\mathbf{A}(\boldsymbol{\mu}) \in \mathbb{R}^{N \times N}$ and the right-hand side forcing 
vector $\mathbf{b}(\boldsymbol{\mu}) \in \mathbb{R}^N$ are parameter-dependent operators governing the high-dimensional physical 
state vector $\mathbf{x}(\boldsymbol{\mu}) \in \mathbb{R}^N$.

By adopting the right ROB $\mathbf{V} \in \mathbb{R}^{N \times n}$ and left test basis $\mathbf{W} \in \mathbb{R}^{N \times n}$ 
established in Section~\ref{sec:framework}, the unaugmented baseline approximation ($\bar{\mathbf{x}}(\boldsymbol{\mu}) = \mathbf{0}$)
collapses the general Petrov-Galerkin PROM equation to the linear algebraic system:
\begin{equation}
	\left(\mathbf{W}^T \mathbf{A}(\boldsymbol{\mu}) \mathbf{V}\right)
	\mathbf{y}(\boldsymbol{\mu}) = \mathbf{W}^T \left( \mathbf{b}(\boldsymbol{\mu}) - \mathbf{A}(\boldsymbol{\mu})
	\mathbf{x}_{\mathrm{ref}} \right)
	\label{eq:petrov_galerkin}
\end{equation}
Here, we assume for simplicity that the Proper Orthogonal Decomposition (POD) method of snapshots utilized to construct $\mathbf{V}$ 
is performed without truncation, fixing the basis dimension $n$ to match the total number of collected training snapshot states $n_s$.
This configuration is well-justified by the static nature of the underlying problem; unlike transient settings that naturally yield 
dense, time-resolved data arrays requiring aggressive compression, static parametric analyses are typically bounded by a limited 
number of snapshot generation opportunities, making the retention of full snapshot energy both computationally tractable and 
preferred to maximize information retention.

In the subsequent subsections, we leverage this linear instantiation to mathematically expose the structural pathologies underlying 
conventional projection operators and detail how our state-shifting strategy restores parametric smoothness.

\subsection{Theoretical origin of intrinsic snapping}
\label{sec:theoretical}

To mathematically characterize the structural pathology under investigation, we first introduce a formal definition of the snapping 
phenomenon within the parametric linear framework. Let $\operatorname{Range}(\mathbf{V}) = \Span \{\mathbf{v}_1, \mathbf{v}_2, \dots, 
\mathbf{v}_{n_s}\}$ denote the subspace spanned by the columns $\mathbf{v}_j$ of the right ROB $\mathbf{V}$. We state that the 
solution of the Petrov-Galerkin PROM~\eqref{eq:petrov_galerkin} \textit{snaps} at a query parameter $\boldsymbol{\mu}^\star \in 
\mathcal{D}$ (where $\boldsymbol{\mu}^{\star} \notin \{ \boldsymbol{\mu}_j \}_{j=1}^{n_s}$) if there exists a training parameter 
$\boldsymbol{\mu}_i \in \{\boldsymbol{\mu}_j\}_{j=1}^{n_s}$ corresponding to a snapshot state in $\operatorname{Range}(\mathbf{V})$ 
(with $\boldsymbol{\mu}_i \neq \boldsymbol{\mu}^\star$) such that the reconstructed approximate solution identically matches the 
training snapshot:
\begin{equation*}
	\tilde{\mathbf{x}}(\boldsymbol{\mu}^{\star}) = \mathbf{x}\left(\boldsymbol{\mu}_i\right)
\end{equation*}

From this definition, it follows that while the simultaneous satisfaction of a non-singular system matrix and a zero projection error 
guarantees exact physical recovery, the snapping pathology can manifest across three distinct operational regimes where these 
conditions are violated:
\begin{enumerate}[label=(\arabic*)]
	\item \textit{Non-singular matrix with non-zero projection error:} The reduced operator remains invertible, but misaligned 
		projection mechanics pathologically draw and lock the unique online coordinate vector onto an incorrect offline 
		training snapshot state.
	\item \textit{Singular matrix with zero projection error:} The exact solution resides within the subspace, but the operator 
		loses uniqueness, creating a non-trivial null space that allows the solution to collapse onto a training snapshot 
		coordinate vector.
\item \textit{Singular matrix with non-zero projection error:} The system simultaneously lacks both a representation of the true state
	and a unique inversion pathway, compounding the structural vulnerability and triggering a breakdown of parametric smoothness.
\end{enumerate}

With these conditions established, we state and prove the following proposition. While the exact recovery of a training snapshot is 
often taken for granted as a self-evident or universally desirable property of an interpolative surrogate, we demonstrate here how 
this rigid identity underscores the mathematical onset of the snapping pathology when an online state solution encounters an 
unmitigated or ill-conditioned reduced operator subspace.

\begin{proposition}\label{pro:p1}
	Consider the parametric linear HDM~\eqref{eq:linear_hdm} reduced via the unaugmented $(\bar{\mathbf{x}}(\boldsymbol{\mu}) = 
	\mathbf{0})$ Petrov-Galerkin projection scheme~\eqref{eq:petrov_galerkin}. If at an online parameter 
	query~$\boldsymbol{\mu}^\star \in \mathcal{D}$, the exact high-dimensional solution can be represented perfectly by
	the affine trial space, such that~$(\mathbf{x}(\boldsymbol{\mu}^\star) - \mathbf{x}_{\mathrm{ref}}) \in \operatorname{Range}
	(\mathbf{V})$, and the corresponding reduced-order system matrix~$\mathbf{W}^T \mathbf{A}(\boldsymbol{\mu}^\star) \mathbf{V}$ 
	is non-singular, then the reconstructed solution from the PROM matches the exact solution perfectly, 
	yielding~$\tilde{\mathbf{x}}(\boldsymbol{\mu}^\star) = \mathbf{x}(\boldsymbol{\mu}^\star)$.
\end{proposition}

\begin{proof}
	Because $\mathbf{x}(\boldsymbol{\mu}^\star) - \mathbf{x}_{\mathrm{ref}} \in \operatorname{Range}(\mathbf{V})$, the exact 
	state resides perfectly within the affine trial space, admitting the exact representation $\mathbf{x}(\boldsymbol{\mu}^\star) 
	= \mathbf{V}\mathbf{y}_i + \mathbf{x}_{\mathrm{ref}}$ for some unique coordinate vector $\mathbf{y}_i \in \mathbb{R}^n$. 
	Premultiplying the full parametric system $\mathbf{A}(\boldsymbol{\mu}^\star)\left(\mathbf{V}\mathbf{y}_i + 
	\mathbf{x}_{\mathrm{ref}}\right) = \mathbf{b}(\boldsymbol{\mu}^\star)$ by the left test basis $\mathbf{W}^T$ yields the 
	projected exact relation:
	\begin{equation*}
		\left(\mathbf{W}^T \mathbf{A}(\boldsymbol{\mu}^\star) \mathbf{V}\right) \mathbf{y}_i = \mathbf{W}^T \left(\mathbf{b}
		(\boldsymbol{\mu}^\star) - \mathbf{A}(\boldsymbol{\mu}^\star) \mathbf{x}_{\mathrm{ref}} \right)
	\end{equation*}
	Subtracting this relation from the online Petrov-Galerkin reduced-order problem~\eqref{eq:petrov_galerkin} isolates the 
	coordinate error vector within the null space of the projected operator:
	\begin{equation*}
		\left(\mathbf{W}^T \mathbf{A}(\boldsymbol{\mu}^\star) \mathbf{V}\right) \left( \mathbf{y}(\boldsymbol{\mu}^\star) - 
		\mathbf{y}_i \right) = \mathbf{0} \end{equation*} 
		Since the reduced-order system matrix $\mathbf{W}^T \mathbf{A}(\boldsymbol{\mu}^\star) \mathbf{V}$ is non-singular by 
		hypothesis, its null space is trivial, uniquely forcing $\mathbf{y}(\boldsymbol{\mu}^\star) = \mathbf{y}_i$. It 
		immediately follows from the unaugmented affine reconstruction mapping -- wherein $\bar{\mathbf{x}} 
		(\boldsymbol{\mu}^\star) = \mathbf{0}$ -- that $\tilde{\mathbf{x}}(\boldsymbol{\mu}^\star) = \mathbf{V}\mathbf{y}
		(\boldsymbol{\mu}^\star) + \mathbf{x}_{\mathrm{ref}} = \mathbf{V}\mathbf{y}_i + \mathbf{x}_{\mathrm{ref}} = 
		\mathbf{x}(\boldsymbol{\mu}^\star)$, completing the proof.
\end{proof}

From an engineering standpoint, the assumption of a non-singular system matrix $\mathbf{W}^T \mathbf{A}(\boldsymbol{\mu}^\star)
\mathbf{V}$ represents a critical point of vulnerability in standard production workflows. Many practical high-fidelity applications 
-- such as structural analysis models of an aircraft undergoing free-flight maneuvers without fixed boundary conditions -- possess 
true or numerical rigid-body modes that render the underlying operator singular or pathologically ill-conditioned. When this uniqueness
pathway is compromised, the exact containment of a physical state within the trial space no longer guarantees exact physical recovery;
instead, the non-trivial or poorly conditioned operator kernel allows the online solver to collapse and lock onto a training 
coordinate vector, triggering a severe breakdown of parametric smoothness.

\begin{remark}[Generality to nonlinear systems]
	\label{rem:nonlinear_prop31}
	Although Proposition~\ref{pro:p1} is formalized above within the linear regime for expository clarity, its core mechanism is 
	fundamentally geometric rather than algebraic. Because the exact recovery property depends strictly on the exact containment 
	of the high-dimensional physical state within the trial space ansatz and the local invertibility of the projected equations, 
	this identity holds true for general nonlinear residual operators $\mathbf{r}(\mathbf{x}; \boldsymbol{\mu}) = \mathbf{0}$. 
	That is, if the exact nonlinear state is perfectly representable by the shifted affine trial space at a query parameter 
	$\boldsymbol{\mu}^\star$, a non-singular projected Jacobian guarantees that the online reduced-order model will uniquely and 
	perfectly reconstruct the high-fidelity solution.
\end{remark}

While Proposition~\ref{pro:p1} establishes how the exact containment of a state within the snapshot subspace leads to a collapse onto 
a training snapshot, the snapping pathology can also be driven by structural insensitivity within the projected operators themselves. 
Specifically, when the left test space filters out parameter-dependent variations in the governing equations, the PROM becomes blind 
to changes in the parameter domain, causing the solution to stagnate or ``snap'' to a constant state across distinct parameter queries.
This mechanism is formalized in the following proposition.

\begin{proposition}\label{pro:p2}
	Consider again the parametric linear HDM~\eqref{eq:linear_hdm} within an affine reconstruction framework. Consider also two 
	queried parameter vectors
	$\boldsymbol{\mu}_1^{\star}$ and $\boldsymbol{\mu}_2^{\star}$ in $\mathcal{D}$. If the following four conditions hold:
	\begin{enumerate}[label=(\arabic*)]
		\item $\mathbf{W}^T \left( \mathbf{b}(\boldsymbol{\mu}_2^{\star}) - \mathbf{b}(\boldsymbol{\mu}_1^{\star}) \right) = 
			\mathbf{0}$,
		\item $\mathbf{W}^T \left( \mathbf{A}(\boldsymbol{\mu}_2^{\star}) - \mathbf{A}(\boldsymbol{\mu}_1^{\star}) \right) 
			\mathbf{V} = \mathbf{0}$,
		\item $\mathbf{W}^T \left( \mathbf{A}(\boldsymbol{\mu}_2^{\star}) - \mathbf{A}(\boldsymbol{\mu}_1^{\star}) \right) 
			\mathbf{x}_{\mathrm{ref}} = \mathbf{0}$,
		\item $\mathbf{W}^T \mathbf{A}(\boldsymbol{\mu}_1^{\star}) \mathbf{V}$ is invertible,
	\end{enumerate}
	then the reconstructed solutions from the PROM are identical, yielding $\tilde{\mathbf{x}}(\boldsymbol{\mu}_2^{\star}) = 
	\tilde{\mathbf{x}}(\boldsymbol{\mu}_1^{\star})$.
\end{proposition}

\begin{proof}
	We evaluate the unaugmented Petrov-Galerkin PROM system~\eqref{eq:petrov_galerkin} at the online parameter 
	query~$\boldsymbol{\mu}_2^\star$:
	\begin{equation*}
		\left(\mathbf{W}^T\mathbf{A}(\boldsymbol{\mu}_2^\star)\mathbf{V}\right)\mathbf{y}(\boldsymbol{\mu}_2^\star) = 
		\mathbf{W}^T\left(\mathbf{b}(\boldsymbol{\mu}_2^\star) - \mathbf{A}(\boldsymbol{\mu}_2^\star)
		\mathbf{x}_{\mathrm{ref}}\right)
	\end{equation*}
	Invoking Conditions~(1), (2), and (3) allows for the direct substitution of the projected operators 
	at~$\boldsymbol{\mu}_2^\star$ with their equivalent expressions evaluated at ~$\boldsymbol{\mu}_1^\star$, simplifying the 
	system directly to:
	\begin{equation*}
		\left(\mathbf{W}^T\mathbf{A}(\boldsymbol{\mu}_1^\star)\mathbf{V}\right)\mathbf{y}(\boldsymbol{\mu}_2^\star) = \mathbf{W}^T
		\left(\mathbf{b}(\boldsymbol{\mu}_1^\star) - \mathbf{A}(\boldsymbol{\mu}_1^\star)\mathbf{x}_{\mathrm{ref}}\right)
	\end{equation*}
	The right-hand side of this expression matches the definition of the reduced-order problem evaluated 
	at~$\boldsymbol{\mu}_1^\star$, yielding:
	\begin{equation*}
		\left(\mathbf{W}^T\mathbf{A}(\boldsymbol{\mu}_1^\star)\mathbf{V}\right)\mathbf{y}(\boldsymbol{\mu}_2^\star) = 
		\left(\mathbf{W}^T\mathbf{A}(\boldsymbol{\mu}_1^\star)\mathbf{V}\right)\mathbf{y} (\boldsymbol{\mu}_1^\star) 
	\end{equation*}
	Rearranging terms groups the coordinate vectors within the null space of the projected operator:
	\begin{equation*}
		\left(\mathbf{W}^T\mathbf{A}(\boldsymbol{\mu}_1^\star)\mathbf{V}\right)\left(\mathbf{y}(\boldsymbol{\mu}_2^\star) - 
		\mathbf{y}(\boldsymbol{\mu}_1^\star)\right) = \mathbf{0}
	\end{equation*}
	By Condition~(4), the reduced-order matrix~$\mathbf{W}^T\mathbf{A}(\boldsymbol{\mu}_1^\star)\mathbf{V}$ is invertible. Its 
	null space is therefore trivial, uniquely forcing~$\mathbf{y}(\boldsymbol{\mu}_2^\star) = \mathbf{y}
	(\boldsymbol{\mu}_1^\star)$. It follows directly from the affine reconstruction mapping that~$\tilde{\mathbf{x}} 
	(\boldsymbol{\mu}_2^\star) = \tilde{\mathbf{x}}(\boldsymbol{\mu}_1^\star)$, completing the proof.
\end{proof}

By setting $\boldsymbol{\mu}_1^{\star} \in \{\boldsymbol{\mu}_j\}_{j=1}^{n_s}$ as a training parameter and $\boldsymbol{\mu}_2^{\star}
\notin \{\boldsymbol{\mu}_j\}_{j=1}^{n_s}$ as an unseen online query, this proposition directly demonstrates a mathematical mechanism 
behind the snapping pathology. Importantly, Proposition~\ref{pro:p2} establishes a \textit{sufficient condition} for the onset of 
intrinsic snapping; while satisfying these structural criteria guarantees a localized loss of parametric sensitivity, it does not
preclude the pathology from manifesting under alternative analytical conditions. This highlights a fundamental structural driver of 
intrinsic snapping wherein the test space underlying the left ROB $\mathbf{W}$ is orthogonal to the parameter-induced variations of 
the high-dimensional operators, leaving the online PROM insensitive to parameter changes along that specific manifold. Consequently, 
if $\boldsymbol{\mu}_1^{\star}$ corresponds to an offline training snapshot location, the model will rigidly snap and return the same
snapshot solution for $\boldsymbol{\mu}_2^{\star}$, inhibiting smooth interpolation and giving rise to a pathologically localized, 
step-like response across that region of the parameter domain.

The structural insensitivity detailed in Proposition~\ref{pro:p2} relies on a strict orthogonality condition that holds uniformly for 
the entire reduced trial subspace. That is, the parameter-induced operator variations are filtered out regardless of the online state.
However, a weaker, state-dependent sufficient condition can be formulated where the operator perturbations are filtered out solely when
evaluated along the specific trajectory of the online solution vector~$\mathbf{y}(\boldsymbol{\mu}_2^\star)$. We transition now to 
this alternative condition for the onset of the pathology, which is formalized in the following proposition.

\begin{proposition}\label{pro:p3}
	Consider again the parametric linear HDM~\eqref{eq:linear_hdm} within an affine reconstruction framework. Consider also two 
	queried parameter vectors $\boldsymbol{\mu}_1^{\star}$ and $\boldsymbol{\mu}_2^{\star}$ in $\mathcal{D}$. If the following 
	four conditions hold:
	\begin{enumerate}[label=(\arabic*)]
		\item $\mathbf{W}^T \left( \mathbf{b}(\boldsymbol{\mu}_2^{\star}) - \mathbf{b}(\boldsymbol{\mu}_1^{\star}) \right) = 
			\mathbf{0}$,
		\item $\mathbf{W}^T \left( \mathbf{A}(\boldsymbol{\mu}_2^{\star}) - \mathbf{A}(\boldsymbol{\mu}_1^{\star}) \right) 
			\mathbf{V} \mathbf{y}(\boldsymbol{\mu}_2^{\star}) = \mathbf{0}$,
		\item $\mathbf{W}^T \left( \mathbf{A}(\boldsymbol{\mu}_2^{\star}) - \mathbf{A}(\boldsymbol{\mu}_1^{\star}) \right) 
			\mathbf{x}_{\mathrm{ref}} = \mathbf{0}$,
		\item $\mathbf{W}^T \mathbf{A}(\boldsymbol{\mu}_1^{\star}) \mathbf{V}$ is invertible,
	\end{enumerate}
	then the reconstructed solutions from the PROM are identical, yielding $\tilde{\mathbf{x}}(\boldsymbol{\mu}_2^{\star}) = 
	\tilde{\mathbf{x}}(\boldsymbol{\mu}_1^{\star})$.
\end{proposition}

\begin{proof}
	We evaluate the unaugmented Petrov-Galerkin PROM system~\eqref{eq:petrov_galerkin} at the online parameter 
	query~$\boldsymbol{\mu}_2^\star$:
	\begin{equation*}
		\left(\mathbf{W}^T\mathbf{A}(\boldsymbol{\mu}_2^\star)\mathbf{V}\right)\mathbf{y}(\boldsymbol{\mu}_2^\star) = 
		\mathbf{W}^T\left(\mathbf{b}(\boldsymbol{\mu}_2^\star) - \mathbf{A}(\boldsymbol{\mu}_2^\star)
		\mathbf{x}_{\mathrm{ref}}\right)
	\end{equation*}
	By invoking the state-dependent constraint in Condition~(2) alongside the forcing and reference state variations from 
	Conditions~(1) and (3), the left- and right-hand sides can be directly substituted with their equivalent projected operators 
	evaluated at~$\boldsymbol{\mu}_1^\star$. This collapses the expression to:
	\begin{equation*}
		\left(\mathbf{W}^T\mathbf{A}(\boldsymbol{\mu}_1^\star)\mathbf{V}\right)\mathbf{y}(\boldsymbol{\mu}_2^\star) = 
		\left(\mathbf{W}^T\mathbf{A}(\boldsymbol{\mu}_1^\star)\mathbf{V}\right)\mathbf{y}(\boldsymbol{\mu}_1^\star)
	\end{equation*}
	Because this resulting system matches the identical algebraic structure obtained in the proof of Proposition~\ref{pro:p2}, 
	the remainder of the proof follows analogously via the invertibility of~$\mathbf{W}^T\mathbf{A}(\boldsymbol{\mu}_1^\star)
	\mathbf{V}$ from Condition~(4), yields~$\mathbf{y}(\boldsymbol{\mu}_2^\star) = \mathbf{y}(\boldsymbol{\mu}_1^\star)$, and 
	consequently dictates that~$\tilde{\mathbf{x}}(\boldsymbol{\mu}_2^\star) = \tilde{\mathbf{x}}(\boldsymbol{\mu}_1^\star)$.
\end{proof}

Proposition~\ref{pro:p3} establishes yet another independent sufficient condition for the onset of intrinsic snapping. The 
fundamental distinction between this formulation and Proposition~\ref{pro:p2} lies in the localization of the matrix 
perturbation constraint. While Proposition~\ref{pro:p2} demands a global orthogonality wherein the test space $\mathbf{W}$ filters out
variations across the entire trial subspace $\operatorname{Range}(\mathbf{V})$, Proposition~\ref{pro:p3} relaxes this requirement, 
demonstrating that snapping will occur if the operator variations are filtered out solely along the specific trajectory of the online
coordinate solution vector $\mathbf{y}(\boldsymbol{\mu}_2^{\star})$. Furthermore, because $\mathbf{x}_{\mathrm{ref}}$ is an arbitrary 
vector outside the range of $\mathbf{V}$, Condition~(3) is required explicitly within the affine setting to cleanly decouple the 
reference state from parameter variations.

\begin{proposition}[Necessary and sufficient condition for stagnation under constant operators]
	\label{pro:p4}
	Consider again the parametric linear HDM~\eqref{eq:linear_hdm} within an affine reconstruction framework, where the system 
	operator is parameter-independent such that $\mathbf{A}(\boldsymbol{\mu}) = \mathbf{A}$. Let $\boldsymbol{\mu}_1^\star, 
	\boldsymbol{\mu}_2^\star \in \mathcal{D}$ be two parameter queries, and assume that $\mathbf{W}^T \mathbf{A} \mathbf{V}$ is 
	invertible. Then, the reconstructed solutions from the PROM are identical, yielding $\tilde{\mathbf{x}}
	(\boldsymbol{\mu}_2^\star) = \tilde{\mathbf{x}}(\boldsymbol{\mu}_1^\star)$ (driven by generalized coordinate stagnation 
	$\mathbf{y}(\boldsymbol{\mu}_2^\star) = \mathbf{y}(\boldsymbol{\mu}_1^\star)$), if and only if:
	\begin{equation} 
		\mathbf{W}^T \left(\mathbf{b}(\boldsymbol{\mu}_2^\star) - \mathbf{b}(\boldsymbol{\mu}_1^\star) \right) = \mathbf{0}
	\end{equation}
\end{proposition}

\begin{proof}
	$(\implies \text{Necessity})$ Assume that the online generalized coordinates stagnate such that $\mathbf{y}
	(\boldsymbol{\mu}_2^\star) = \mathbf{y}(\boldsymbol{\mu}_1^\star)$. Evaluating the unaugmented Petrov-Galerkin PROM at 
	$\boldsymbol{\mu}_1^\star$ and $\boldsymbol{\mu}_2^\star$ yields:
	\begin{align*}
		\left(\mathbf{W}^T \mathbf{A} \mathbf{V}\right) \mathbf{y}(\boldsymbol{\mu}_1^\star) &= \mathbf{W}^T \left(
		\mathbf{b}(\boldsymbol{\mu}_1^\star) - \mathbf{A} \mathbf{x}_{\text{ref}} \right) \\
		\left(\mathbf{W}^T \mathbf{A} \mathbf{V}\right) \mathbf{y}(\boldsymbol{\mu}_2^\star) &= \mathbf{W}^T \left( 
		\mathbf{b}(\boldsymbol{\mu}_2^\star) - \mathbf{A} \mathbf{x}_{\text{ref}} \right)
	\end{align*}
	Subtracting the first equation from the second gives:
	\begin{equation*}
		\left(\mathbf{W}^T \mathbf{A} \mathbf{V}\right) \left( \mathbf{y}(\boldsymbol{\mu}_2^\star) - \mathbf{y}(\boldsymbol{\mu}_1^\star) 
		\right) = \mathbf{W}^T\left( \mathbf{b}(\boldsymbol{\mu}_2^\star) - \mathbf{b}(\boldsymbol{\mu}_1^\star) \right)
	\end{equation*}
	Substituting $\mathbf{y}(\boldsymbol{\mu}_2^\star) - \mathbf{y}(\boldsymbol{\mu}_1^\star) = \mathbf{0}$ into the left-hand 
	side directly yields:
	\begin{equation*}
		\mathbf{W}^T \left( \mathbf{b}(\boldsymbol{\mu}_2^\star) - \mathbf{b}(\boldsymbol{\mu}_1^\star) \right) = \mathbf{0}
	\end{equation*}

	$(\impliedby \text{Sufficiency})$ Conversely, assume that $\mathbf{W}^T \left( \mathbf{b}(\boldsymbol{\mu}_2^\star) - 
	\mathbf{b}(\boldsymbol{\mu}_1^\star) \right) =\mathbf{0}$. The Petrov-Galerkin reduced-order system evaluated at 
	$\boldsymbol{\mu}_2^\star$ expands to:
	\begin{equation*}
		\left(\mathbf{W}^T \mathbf{A} \mathbf{V}\right) \mathbf{y}(\boldsymbol{\mu}_2^\star) = \mathbf{W}^T 
		\left( \mathbf{b}(\boldsymbol{\mu}_2^\star) - \mathbf{A} \mathbf{x}_{\text{ref}} \right) = \mathbf{W}^T 
		\left( \mathbf{b}(\boldsymbol{\mu}_1^\star) - \mathbf{A} \mathbf{x}_{\text{ref}} \right) = \left(\mathbf{W}^T 
		\mathbf{A} \mathbf{V}\right) \mathbf{y} (\boldsymbol{\mu}_1^\star)
	\end{equation*}
	Grouping terms isolates the coordinate difference within the null space of the projected operator:
	\begin{equation*}
		\left(\mathbf{W}^T \mathbf{A} \mathbf{V}\right) \left( \mathbf{y}(\boldsymbol{\mu}_2^\star) - \mathbf{y}
		(\boldsymbol{\mu}_1^\star) \right) = \mathbf{0}
	\end{equation*}
	Because $\mathbf{W}^T \mathbf{A} \mathbf{V}$ is invertible by hypothesis, its null space is trivial, uniquely forcing 
	$\mathbf{y}(\boldsymbol{\mu}_2^\star) = \mathbf{y}(\boldsymbol{\mu}_1^\star)$. It follows immediately from the unaugmented 
	affine reconstruction mapping that $\tilde{\mathbf{x}}(\boldsymbol{\mu}_2^\star) =\mathbf{V}\mathbf{y}
	(\boldsymbol{\mu}_2^\star) + \mathbf{x}_{\text{ref}} = \mathbf{V}\mathbf{y}(\boldsymbol{\mu}_1^\star) + 
	\mathbf{x}_{\text{ref}} = \tilde{\mathbf{x}}(\boldsymbol{\mu}_1^\star)$, completing the proof.
\end{proof}

Proposition~\ref{pro:p4} highlights a severely pessimistic outlook for baseline PMOR configurations operating on constant-operator 
systems. Because the left test basis is highly rectangular ($\mathbf{W} \in \mathbb{R}^{N \times n}$ with $n \ll N$), its transpose possesses 
an immense null space of dimension $N-n$. Elevating the projected forcing constraint $\mathbf{W}^T \Delta \mathbf{b} = \mathbf{0}$ to 
an if-and-only-if equivalence demonstrates that preventing generalized coordinate stagnation hinges entirely on the low-dimensional 
trajectory of high-dimensional physical forcing variations avoiding this multi-million-dimensional numerical blind spot -- a structural
requirement that is routinely violated in practical operational settings.

\subsection{Mitigation strategy for intrinsic snapping}
\label{subsec:mitigation_linear}

Driven by the theoretical limits established in the preceding section, we evaluate this mitigation framework within the most 
algebraically severe environment: the parameter-independent operator configuration ($\mathbf{A}(\boldsymbol{\mu}) = \mathbf{A}$). This
scoping choice is mathematically deliberate; because the constant operator regime strips away the model's algebraic defenses and 
exposes it to a singular point of failure, it represents the definitive structural bottleneck for linear PROM stability. Rather than attempting to restrict the physical trajectory of the parameter-dependent forcing vector $\mathbf{b}(\boldsymbol{\mu})$ -- which is 
dictated entirely by the underlying application physics -- the strategy proposed herein focuses on actively altering the projection 
mechanics. By structurally forcing a break in the orthogonality condition $\mathbf{W}^T \Delta \mathbf{b} = \mathbf{0}$, we ensure 
that the reduced-order operator retains dynamic sensitivity to parametric variations across the entire domain $\mathcal{D}$, thereby 
systematically eliminating the mathematical driver behind pathologically rigid snapshot latching. By proving that the framework 
restores continuous sensitivity under thisworst-case scenario, we establish a definitive lower bound for its efficacy across more 
general linear regimes.

\subsubsection{Affine state augmentation via a parameter-dependent shift vector}

To structurally prevent the parametric variations of the linear forcing vector from disappearing into the null space of the left test 
basis, we instantiate the general projected PROM system~\eqref{eq:petrov_galerkin} under the parameter-independent operator 
configuration ($\mathbf{A}(\boldsymbol{\mu}) = \mathbf{A}$). By directly applying the abstract shifted projection expression, the 
augmented linear PROM simplifies to:
\begin{equation}
    \left(\mathbf{W}^T\mathbf{A}\mathbf{V}\right)\mathbf{y}(\boldsymbol{\mu}) = \mathbf{W}^T\left(\mathbf{b}(\boldsymbol{\mu}) - \mathbf{A}\mathbf{x}_{\mathrm{ref}} - \mathbf{A}\bar{\mathbf{x}}(\boldsymbol{\mu})\right)
	\label{eq:augmented_prom_linear}
\end{equation}
The appearance of the projected shift term $\mathbf{W}^T\mathbf{A}\bar{\mathbf{x}}(\boldsymbol{\mu})$ on the right-hand side directly 
alters the mapping between the parameter domain and the projected forcing space, offering an explicit algebraic mechanism to 
circumvent the intrinsic snapping limits of the baseline framework.

Let $\{\mathbf{x}(\boldsymbol{\mu}_j)\}_{j=1}^{n_s}$ represent the set of high-dimensional physical state snapshots computed offline 
at the training locations, corresponding to the high-fidelity forcing vectors $\{\mathbf{b}(\boldsymbol{\mu}_j)\}_{j=1}^{n_s}$. In this
linear regime, the parameter-dependent shift vector $\bar{\mathbf{x}}(\boldsymbol{\mu})$ is defined via an explicit algebraic map 
determined by the parametric variations of the high-fidelity forcing vector:
\begin{equation}
	\bar{\mathbf{x}}(\boldsymbol{\mu}) = \mathbf{A}^{-1} \left( \mathbf{b}(\boldsymbol{\mu}) - \sum_{j=1}^{n_s} 
	\phi_j(\boldsymbol{\mu}) \mathbf{b}(\boldsymbol{\mu}_j) \right)
	\label{eq:forcing_vector}
\end{equation}
where $\phi_j(\boldsymbol{\mu}) : \mathbb{R}^{N_\mu} \to \mathbb{R}$ are multi-dimensional, scalar-valued parameter space interpolants
(such as Radial Basis Functions) structurally constrained to satisfy the standard Kronecker $\delta$-property at the training 
coordinates, yielding $\phi_j(\boldsymbol{\mu}_i) = \delta_{ji}$ for any $i,j \in \{1, \dots, n_s\}$.

It is worth noting that while Equation~\eqref{eq:forcing_vector} introduces the high-dimensional inverse matrix $\mathbf{A}^{-1}$ for 
theoretical completeness, its evaluation does not entail an expensive online high-dimensional linear solve. Distributing the left 
projection operators across the shift vector expression yields:
\begin{equation}
	\mathbf{W}^T\mathbf{A}\bar{\mathbf{x}}(\boldsymbol{\mu}) = \mathbf{W}^T\mathbf{b}(\boldsymbol{\mu}) - \sum_{j=1}^{n_s} 
	\phi_j(\boldsymbol{\mu}) \mathbf{W}^T\mathbf{b} (\boldsymbol{\mu}_j)
	\label{eq:projected_shift_evaluation}
\end{equation}
Because the composition $\mathbf{A}\mathbf{A}^{-1}$ evaluates identically to the identity operator prior to projection, the 
high-dimensional matrix inversion is analytically bypassed online. Substituting this simplified 
relation~\eqref{eq:projected_shift_evaluation} directly back into the augmented PROM system~\eqref{eq:augmented_prom_linear} reveals 
that the exact system of equations solved online simplifies to:
\begin{equation*}
	\left(\mathbf{W}^T\mathbf{A}\mathbf{V}\right)\mathbf{y}(\boldsymbol{\mu}) = \sum_{j=1}^{n_s} \phi_j(\boldsymbol{\mu}) 
	\mathbf{c}_j - \mathbf{W}^T\mathbf{A} \mathbf{x}_{\mathrm{ref}}
	\label{eq:final_online_assembly}
\end{equation*}
where the constant, low-dimensional vectors~$\mathbf{c}_j \in \mathbb{R}^n$ are pre-assembled once offline according to:
\begin{equation}
	\mathbf{c}_j = \mathbf{W}^T\mathbf{b}(\boldsymbol{\mu}_j)
\end{equation}
Consequently, evaluating the effect of the shift vector online requires no high-dimensional operations whatsoever. The online 
realization proceeds in three trivial steps: (i) evaluate the scalar interpolants~$\phi_j(\boldsymbol{\mu})$ at the query 
parameter~$\boldsymbol{\mu}$, (ii) compute the linear combination of the pre-assembled low-dimensional vectors to form the 
reconstructed forcing state $\mathbf{b}_{sh} = \sum_{j=1}^{n_s} \phi_j(\boldsymbol{\mu}) \mathbf{c}_j$, and (iii) solve the remaining 
$n \times n$ dense linear system for the coordinates~$\mathbf{y}(\boldsymbol{\mu})$.

This explicit construction injects the missing spatial information back into the trial space representation prior to projection, 
successfully breaking the test-space orthogonality blind spot online while requiring no spatial PDE solves during the offline setup 
phase. Critically, this structural perturbation modifies only the kinematics of the reduced-order trial space representation; it does
not alter the underlying physics or the high-fidelity problem of interest. Because the transformation is algebraic and evaluated 
directly within the residual operator, the exact solution manifold remains invariant, ensuring that near-field fidelity targets are 
uncompromised. As will be rigorously demonstrated in the next subsection, this structural perturbation ensures the model retains 
dynamic sensitivity across the entire parameter domain without introducing modeling bias.

\subsubsection{Conditions for breaking subspace orthogonality}
\label{sec:breaking_orthogonality}

To demonstrate how the augmented formulation~\eqref{eq:augmented_prom_linear} systematically eliminates the snapping pathology, we 
evaluate the conditions for generalized coordinate stagnation when an unseen online parameter $\boldsymbol{\mu}_2^\star = 
\boldsymbol{\mu}^\star$ is queried, yielding a new high-fidelity forcing vector $\mathbf{b}(\boldsymbol{\mu}_2^\star)$. We analyze the
system perturbations relative to an offline training snapshot location $\boldsymbol{\mu}_1^\star = \boldsymbol{\mu}_i$, where by 
virtue of the snapshot reproduction property, $\bar{\mathbf{x}} (\boldsymbol{\mu}_1^\star) = \mathbf{0}$. Let $\Delta \mathbf{b} = 
\mathbf{b}(\boldsymbol{\mu}_2^\star) - \mathbf{b} (\boldsymbol{\mu}_1^\star)$ and $\Delta \bar{\mathbf{x}} = \bar{\mathbf{x}}
(\boldsymbol{\mu}_2^\star) - \bar{\mathbf{x}}(\boldsymbol{\mu}_1^\star) = \bar{\mathbf{x}} (\boldsymbol{\mu}_2^\star)$.

By tracking the difference between the projected equations under the augmented framework, the condition required to force a 
pathologically rigid generalized coordinate stagnation simplifies to:
\begin{equation}
	\mathbf{W}^T(\Delta \mathbf{b} - \mathbf{A}\Delta\bar{\mathbf{x}}) = \mathbf{0}
	\label{eq:stagnation_condition}
\end{equation}
In a baseline linear PROM where $\bar{\mathbf{x}}(\boldsymbol{\mu}) = \mathbf{0}$, snapping is mathematically guaranteed the moment 
the physical forcing perturbation $\Delta \mathbf{b}$ aligns with the $(N-n)$-dimensional null space of $\mathbf{W}^T$. However, 
substituting the dynamically constructed forcing map~\eqref{eq:forcing_vector} into~\eqref{eq:stagnation_condition} yields:
\begin{equation*}
	\mathbf{W}^T \left( \left(\mathbf{b}(\boldsymbol{\mu}_2^\star) - \mathbf{b}(\boldsymbol{\mu}_1^\star)\right) - 
	\mathbf{A}\left[ \mathbf{A}^{-1}\left( \mathbf{b}(\boldsymbol{\mu}_2^\star) - \sum_{j=1}^{n_s} 
	\phi_j(\boldsymbol{\mu}_2^\star) \mathbf{b}(\boldsymbol{\mu}_j) \right) \right] \right) = \mathbf{0}
\end{equation*}
Because the high-dimensional operator composition $\mathbf{A}\mathbf{A}^{-1}$ evaluates to the identity matrix prior to projection, 
the online forcing vector $\mathbf{b}(\boldsymbol{\mu}_2^\star)$ cancels out perfectly from the expression. The terms rearrange and 
simplify directly to:
\begin{equation}
	\mathbf{W}^T \left( \sum_{j=1}^{n_s} \phi_j(\boldsymbol{\mu}_2^\star) \mathbf{b}(\boldsymbol{\mu}_j) - 
	\mathbf{b}(\boldsymbol{\mu}_1^\star) \right) = \mathbf{0}
	\label{eq:collapsed_stagnation}
\end{equation}

Equation~\eqref{eq:collapsed_stagnation} reveals that the online forcing vector $\mathbf{b}(\boldsymbol{\mu}_2^\star)$ has been 
entirely eliminated from the projection matrix filter. Because $\mathbf{W}^T \mathbf{A} \mathbf{V}$ is nonsingular, 
Equation~\eqref{eq:collapsed_stagnation} elevates to a necessary and sufficient condition for generalized coordinate stagnation 
($\mathbf{y}(\boldsymbol{\mu}_2^\star) = \mathbf{y}(\boldsymbol{\mu}_1^\star)$) under the state-shifted framework. Consequently, 
preventing snapping requires proving that Equation~\eqref{eq:collapsed_stagnation} fails to hold -- i.e., becomes a strict inequality 
-- for unseen query points $\boldsymbol{\mu}_2^\star \neq \boldsymbol{\mu}_1^\star$.

Because the parameter interpolants $\phi_j(\boldsymbol{\mu})$ are continuous, nonconstant functions satisfying 
$\phi_j(\boldsymbol{\mu}_i) = \delta_{ji}$, the projected interpolant $\sum_{j=1}^{n_s} \phi_j(\boldsymbol{\mu}_2^\star) 
\mathbf{W}^T \mathbf{b}(\boldsymbol{\mu}_j)$ cannot identically evaluate to the constant vector $\mathbf{W}^T 
\mathbf{b}(\boldsymbol{\mu}_1^\star)$ over any open neighborhood in $\mathcal{D}$, provided the offline projected forcing vectors 
$\{\mathbf{W}^T \mathbf{b}(\boldsymbol{\mu}_j)\}_{j=1}^{n_s}$ contain distinct physical information. Thus, any potential entry point 
where Equation~\eqref{eq:collapsed_stagnation} holds is restricted to isolated, measure-zero roots. The invertibility of $\mathbf{W}^T
\mathbf{A} \mathbf{V}$ subsequently guarantees that $y(\boldsymbol{\mu}_2^\star) \neq y(\boldsymbol{\mu}_1^\star)$ for almost all 
$\boldsymbol{\mu}_2^\star \neq \boldsymbol{\mu}_1^\star$, formally ruling out wide-neighborhood coordinate stagnation and proving the 
restoration of continuous parametric sensitivity across the linear regime.

\section{Numerical snapping in PMOR for parametric nonlinear problems}
\label{sec:nonlinear}

In this section, we shift our focus to Type~3 (numerical) snapping, using parametric, nonlinear FV-based external flow CFD problems 
as the primary backdrop for our discussion. We establish the physical mechanism of numerical snapping driven by spatial mesh-resolution
scaling and discretization artifacts.

To contextualize this failure mode, we first briefly review the historical mitigation strategy developed within our research group, 
which was inspired by the insights established in~\cite{washabaugh2014enforcing, washabaugh2016faster}. This strategy dynamically 
toggles between descriptor and non-descriptor formulations based on which boundary conditions are governed by the parameter space.

In a FV CFD framework, the high-fidelity, semi-discrete governing equations can be expressed in descriptor form as:
\begin{equation*}
	\mathbf{\Omega} \frac{\mathrm{d}\mathbf{x}(t; \boldsymbol{\mu})}{\mathrm{d}t} + \mathbf{r}(\mathbf{x}(t; \boldsymbol{\mu}); 
	\boldsymbol{\mu}) = \mathbf{0}
\end{equation*}
where $\mathbf{x}(t; \boldsymbol \mu) \in \mathbb{R}^N$ represents the high-dimensional fluid state vector, $\mathbf{\Omega} \in 
\mathbb{R}^{N \times N}$ is the diagonal matrix of cell volumes (or dual cell control volumes) playing the role of a generalized mass 
matrix, and $\mathbf{r}(\mathbf{x}(t; \boldsymbol{\mu}); \boldsymbol{\mu}) \in \mathbb{R}^N$ is the spatial residual operator 
containing flux balances and source terms. The descriptor form retains the explicit physical volumes alongside the differential 
equations, preserving the unmanipulated structure of the governing physical laws. While our primary interest in this work centers on 
steady-state flows, retaining the transient term here facilitates the direct derivation of the non-descriptor form.

The non-descriptor formulation is obtained by premultiplying the system by the inverse of the cell volume matrix, 
$\mathbf{\Omega}^{-1}$, isolating the state derivative:
\begin{equation*}
	\frac{\mathrm{d}\mathbf{x}(t; \boldsymbol{\mu})}{\mathrm{d}t} + \mathbf{\Omega}^{-1} 
	\mathbf{r}(\mathbf{x}(t; \boldsymbol{\mu}); \boldsymbol{\mu}) = \mathbf{0}
\end{equation*}

Restricting both frameworks to steady-state operations $(\mathrm{d}\mathbf{x}(t; \boldsymbol{\mu})/\mathrm{d}t = \mathbf{0})$, the 
corresponding time-independent residual systems reduce to:
\begin{align}
	\text{\textit{Descriptor form:}} \quad & \mathbf{r}(\mathbf{x}(\boldsymbol{\mu}); \boldsymbol{\mu}) = \mathbf{0} 
	\label{eq:DF}\\
	\text{\textit{Non-descriptor form:}} \quad & \mathbf{\Omega}^{-1} \mathbf{r}(\mathbf{x}(\boldsymbol{\mu}); \boldsymbol{\mu}) 
	= \mathbf{0} \label{eq:NDF}
\end{align}

In the following subsections, we demonstrate how grid-refinement disparities in FV meshes interact with these two steady-state 
formulations to induce or alleviate numerical snapping.

\subsection{Common numerical snapping mechanism in external flow CFD}
\label{sec:numerical_snapping_mechanism}

In external flow CFD applications, the computational fluid domain is partitioned into distinct spatial regions characterized by vastly
different spatial grid resolutions. Typically, the near-field region of the domain is finely discretized to accurately capture 
boundary layers and resolve strong velocity and pressure gradients localized near solid wall boundaries. Conversely, the far-field 
region of the domain is discretized with significantly coarser elements, as the flow solution there is close to uniform, particularly 
when the far-field boundary is positioned at a long distance away from the obstacle and a lower order (first or second-order) absorbing
boundary condition is applied on the far-field boundary.

Consequently, in many CFD codes, by default, the non-descriptor form of the governing discrete equations~\eqref{eq:NDF} is solved in 
practice so that, given a tolerance $\epsilon$ for the convergence of the typically iterative algebraic flow solver, the solution is 
more accurate in the near-field region than if the descriptor form~\eqref{eq:DF} was solved. As we will see below, this 
grid-resolution-dependent equation scaling has severe consequences on Type~3 numerical snapping.

Consider first the case where the parameter space $\mathcal{D}$ contains parameters $\boldsymbol{\mu}$ that affect the far-field 
boundary conditions -- for example, freestream Mach number $M_{\infty}$, angle of attack $\alpha$, and sideslip angle $\beta$ -- but 
does not contain parameters that affect the wall geometry or wall boundary conditions. 

In this case, it should be noted that in the descriptor form~\eqref{eq:DF}, the dependence of $\mathbf{r}(\mathbf{x}(\boldsymbol{\mu});
\boldsymbol{\mu})$ on such parameters $\boldsymbol{\mu}$ in the far-field is explicit. Furthermore, given that the volumes of the CFD 
mesh cells are large in the far-field, this means that essentially the descriptor residual ``sees'' the dependence on those parameters
$\boldsymbol{\mu}$ that govern the freestream conditions. It follows that in general, this formulation strongly suppresses the onset 
of snapping to a solution associated with different far-field boundary conditions, making such a failure mode highly unlikely. This is
because the residual in the far-field becomes particularly dominant for two reasons: discrepancies in the parameters on which this 
residual depends explicitly; and the sheer magnitude of the residual in this region amplified by the large size of the far-field mesh 
cells.

Following the same reasoning for the non-descriptor form \eqref{eq:NDF}, premultiplying the residual by $\mathbf{\Omega}^{-1}$ scales 
the far-field residuals by the inverse of their large cell volumes ($1/\mathbf{\Omega}_{i}$), drastically dampening the explicit 
parameter dependence at the far-field boundary. Simultaneously, the ultra-small volumes of the finely discretized near-field cells 
yield massive inverse scaling factors ($1/\mathbf{\Omega}_{i}$). This creates a severe numerical imbalance where the unperturbed 
near-field residuals are artificially blown up, dominating the global residual norm and masking the explicit parameter variations occurring at the far-field boundary. During online inference, the solver finds that the heavily amplified near-field residual is minimized 
most effectively by locking onto the nearest training snapshot state. Consequently, the non-descriptor form actively promotes snapping
to a solution associated with different freestream boundary conditions.

Finally, consider the converse case where the parameter space $\mathcal{D}$ contains parameters that affect the wall geometry or wall 
boundary conditions, but does not contain parameters that affect the far-field boundary conditions:

\begin{itemize} 
	\item \textit{Descriptor form:} The explicit parameter dependence now resides in the highly refined near-field region. However,
		because near-field cell volumes $\mathbf{\Omega}_{i}$ are extremely small, the unweighted descriptor residual assigns 
		negligible numerical weight to these explicit geometric or other variations. Instead, the unperturbed, massive cell 
		volumes in the far-field dominate the global residual calculation, blinding the online solver to localized surface 
		variations and promoting numerical snapping to a training snapshot.
	\item \textit{Non-descriptor form:} Premultiplying by $\mathbf{\Omega}^{-1}$ scales the near-field geometric residual 
		variations by $1/\mathbf{\Omega}_{i}$. Because these near-field cell volumes are small, their inverse scaling factors 
		are massive, appropriately amplifying the localized geometric updates. This alignment between explicit parameter 
		dependence and high numerical weighting allows the online solver to smoothly track parameter-induced surface 
		deformations, thereby preventing numerical snapping.  
\end{itemize}

It is worth emphasizing that these scaling mechanics are particularly consequential for residual-minimizing PMOR frameworks, such as 
the LSPG method~\cite{carlberg2011efficient, carlberg2013gnat}. Because LSPG determines the online generalized coordinates 
$\mathbf{y}(\boldsymbol{\mu})$ by directly minimizing the $\ell_2$-norm of the high-dimensional residual vector, any severe numerical 
imbalance introduced by spatial mesh grading or formulation scaling ($\mathbf{\Omega}$ vs. $\mathbf{\Omega}^{-1}$) directly distorts 
the online minimization landscape. Specifically, when the parameter space governs far-field conditions, the non-descriptor form causes
the unperturbed near-field residual to dominate the $\ell_2$-norm, driving the online optimization process into pathologically sharp 
local minima centered at the offline training coordinates. Conversely, when the parameter space governs wall geometry, the descriptor 
form causes the unperturbed far-field residual to dominate the $\ell_2$-norm, inducing the exact same locking behavior. In both 
scenarios, the residual imbalance prevents the solver from tracking parameter-induced updates, forcing the model to rigidly snap onto 
an adjacent training snapshot rather than seeking a smooth parametric interpolation.

\subsection{Corresponding mitigation strategies}

\subsubsection{Historical mitigation: limitation and lack of robustness}
\label{sec:historical_mitigation_limitation}

It is important to stress that the Type~3 numerical snapping mechanism detailed in Section~\ref{sec:numerical_snapping_mechanism} is 
by no means a theoretical speculation, but rather a concrete failure mode observed in practical PMOR 
workflows~\cite{washabaugh2014enforcing, washabaugh2016faster}. To mitigate this behavior, historical practice within our research 
group has relied on dynamically toggling between the descriptor~\eqref{eq:DF} and non-descriptor~\eqref{eq:NDF} formulations based on 
the active boundary parameterization. However, this formulation-toggling paradigm operates merely as an empirical heuristic rather 
than a guaranteed safeguard. While it may stochastically suppress snapshot latching in isolated parameter regimes, it offers no 
structural protection against snapping under routine operational constraints, introduces severe near-field fidelity trade-offs, and 
suffers from two fundamental operational limitations:

\begin{enumerate} 
	\item \textit{Dual-parameter deadlock:} As established in Section~\ref{sec:numerical_snapping_mechanism}, suppressing 
		far-field snapping strictly demands the unweighted descriptor form, whereas suppressing near-field wall snapping 
		requires the volume-scaled non-descriptor form. When the parameter space $\mathcal{D}$ simultaneously governs both 
		far-field conditions (e.g., free-stream Mach number $M_{\infty}$ or angle of attack $\alpha$) and near-field wall 
		characteristics (e.g., surface geometry deformations or wall boundary conditions), neither formulation can resolve 
		both failure modes at once. The practitioner is left structurally deadlocked, as formulation-toggling is inherently 
		incapable of accommodating such coupled parameter spaces.  
	\item \textit{Near-field accuracy degradation:} Even in isolated settings where preventing snapping calls exclusively for the 
		descriptor form (e.g., purely far-field parameterizations), adopting this formulation introduces severe accuracy 
		compromises. Because the unweighted descriptor residual assigns negligible numerical weight to small, highly refined 
		near-field cells, an iterative algebraic solver or residual-minimizing PMOR scheme terminating at a standard solver 
		convergence tolerance $\epsilon$ fails to adequately minimize localized boundary layer errors. Consequently, while 
		the descriptor form may prevent the model from snapping onto a training snapshot, the reconstructed solution lacks 
		near-field fidelity. Under routine operational constraints, the resulting PROM does not snap, but it fails to meet 
		accuracy targets.
\end{enumerate}

\subsubsection{Mitigation via a spatially-varying shift vector}

To address both the fundamental structural deadlock and the accuracy compromises inherent to the historical formulation-toggling 
approach, our proposed strategy rests upon two core methodological pillars:
\begin{enumerate}    
	\item \textit{Universal adoption of the non-descriptor form:} We adhere strictly and universally to the non-descriptor 
		formulation~\eqref{eq:NDF}. By retaining the volume-inverse matrix $\mathbf{\Omega}^{-1}$, the underlying 
		high-fidelity model and the corresponding PROM consistently prioritize error minimization in the highly 
		refined near-field cells, thereby unconditionally preserving near-field simulation accuracy and boundary layer 
		fidelity targets.
	\item \textit{Integration of the unified state-shifting framework:} To neutralize the Type~3 numerical snapping vulnerabilities
		inherently induced by the non-descriptor form, we adopt the general state-shifting framework introduced in 
		Section~\ref{sec:framework}. By embedding a parameter-dependent spatial shift vector $\bar{\mathbf{x}}
		(\boldsymbol{\mu})$ directly into the trial space representation prior to projection, explicit parametric transport 
		and exact boundary condition enforcement are restored to the high-dimensional state \textit{a priori}. Consequently, 
		continuous parametric sensitivity is maintained entirely independently of the mathematical formulation (descriptor or 
		non-descriptor) or equation scaling chosen for the governing PDE.
\end{enumerate}

To concretely instantiate this state-shifting strategy within the non-descriptor CFD framework established above, we consider 
external compressible viscous flows governed by the Navier-Stokes equations and discretized over $N_c$ FV cells (focusing on the 
laminar regime for expository simplicity without loss of generality). In this context, the high-dimensional state vector 
$\mathbf{x}(\boldsymbol{\mu}) \in \mathbb{R}^N$ (with $N = n_{\mathrm{vars}} N_c$, where $n_{\mathrm{vars}}$ is the number of 
conserved fluid variables per cell) is explicitly defined as the concatenated array of local conservative variable vectors:
\begin{equation*} 
	\mathbf{x}(\boldsymbol{\mu}) = 
	\begin{bmatrix} 
		\mathbf{x}_1(\boldsymbol{\mu}) \\[1ex]
		\mathbf{x}_2(\boldsymbol{\mu}) \\[1ex]
		\vdots \\[1ex]
		\mathbf{x}_{N_c}(\boldsymbol{\mu}) 
	\end{bmatrix}, \quad \mathrm{with} \quad 
	\mathbf{x}_i(\boldsymbol{\mu}) = \begin{bmatrix} \rho_i(\boldsymbol{\mu}) \\[1ex] (\rho u)_i(\boldsymbol{\mu}) \\[1ex] 
		(\rho v)_i(\boldsymbol{\mu}) \\[1ex] 
		(\rho w)_i(\boldsymbol{\mu}) \\[1ex] (\rho e_t)_i(\boldsymbol{\mu}) \end{bmatrix} 
	\in \mathbb{R}^5, \quad i = 1, \dots, N_c 
	\label{eq:conservative_state_vector}
\end{equation*}
where $\rho_i$ is the fluid density, $(\rho u)_i$, $(\rho v)_i$, $(\rho w)_i$ are the Cartesian momentum components, and $(\rho e_t)_i$
is the total energy per unit volume evaluated within control volume $i$ -- that is,
\begin{equation} 
	e_t = e + \frac{1}{2}\left(u^2 + v^2 + w^2\right)
\end{equation}
where $e$ is the specific internal energy and $\frac{1}{2}\left(u^2 + v^2 + w^2\right)$ represents the specific kinetic energy of the 
fluid.

It is worth noting that within our proposed augmented affine ansatz~\eqref{eq:affine_approx_shifted} -- which embeds an explicit, 
parameter-dependent spatial tracking mechanism directly into the high-dimensional state representation prior to any projection or 
residual minimization step -- each column of the right ROB $\mathbf{V}$ strictly satisfies homogeneous wall boundary conditions and 
non-parametric far-field conditions (evaluating to zero at those boundaries). Consequently, $\mathbf{x}_{\mathrm{ref}}$ and the 
spatial shift vector $\bar{\mathbf{x}}(\boldsymbol{\mu})$ can be systematically constructed such that the reconstructed state 
$\tilde{\mathbf{x}}(\boldsymbol{\mu})$ unconditionally satisfies the physical wall boundary conditions (whether homogeneous or 
inhomogeneous) as well as the far-field boundary conditions, even when the latter are strongly parameter-dependent.

Furthermore, the static reference state $\mathbf{x}_{\mathrm{ref}}$ is inherently continuous across the spatial domain because it is 
constructed directly as a linear combination of physical snapshot solutions, each of which represents an exact, smooth discrete 
solution to the governing Navier-Stokes equations.

Hence, what remains is to construct the dynamic shift vector $\bar{\mathbf{x}}(\boldsymbol{\mu})$ not only to restore explicit 
parametric tracking and prevent numerical snapping, but also to ensure spatial continuity throughout the fluid domain in the presence 
of parametric far-field and parametric wall boundary conditions. Spatial continuity of the shift vector is of paramount importance: 
any artificial slope jump or spatial discontinuity introduced into $\bar{\mathbf{x}}(\boldsymbol{\mu})$ would generate localized 
non-physical flux imbalances and massive artificial residual spikes in the high-dimensional PDE discretization, degrading solver 
convergence and compromising solution fidelity.

To achieve this spatial smoothness, our approach formulates the construction of the shift vector as the solution to a scaled harmonic 
extension problem. While the underlying CFD state is a multi-component vector field, the spatial diffusion of boundary variations can 
be efficiently governed by a single scalar Laplace equation (in the sense of one scalar unknown $s$ per spatial grid node or cell 
center). At the continuous level, this scalar Laplace problem is formulated as:
\begin{equation} 
	\begin{aligned} 
		\nabla \cdot \left( \Omega_{\mathrm{cell}}^p \, \nabla s \right) &= 0 \quad \text{in} \quad \Omega \\ 
		s &= 1 \quad \text{on} \quad \partial \Omega_{\mathrm{freestream}} \\ 
		s &= 0 \quad \text{on} \quad \partial \Omega_{\mathrm{wall}} 
	\end{aligned} 
	\label{eq:laplace_shift_framework} 
\end{equation}
where $\Omega \subset \mathbb{R}^3$ denotes the three-dimensional computational fluid domain, $\partial 
\Omega_{\mathrm{freestream}}$ and $\partial \Omega_{\mathrm{wall}}$ represent the far-field freestream and solid wall boundaries, 
respectively, $\Omega_{\mathrm{cell}}$ represents the local cell volume, and $p \in [0, 1]$ is a grid-grading scaling exponent.

The primary purpose of the cell-volume weight $\Omega_{\mathrm{cell}}^p$ is to directly control the steepness and spatial decay of 
the Laplace solution $s$ within the near-field region. In aerodynamic applications featuring thin boundary layers, the computational 
grid is heavily refined near solid surfaces, resulting in cell volumes $\Omega_{\mathrm{cell}}$ that are several orders of magnitude 
smaller than those in the far field ($\Omega_{\mathrm{cell}} \ll 1$).

While setting $p = 0$ yields a uniform Laplace operator where freestream variations penetrate deeply toward the wall, choosing $p$ 
closer to $1$ drastically reduces the local diffusion coefficient $\Omega_{\mathrm{cell}}^p$ inside the ultra-small near-field cells. 
This effectively shields the viscous region, forcing $s$ to decay rapidly outside the boundary layer and remain virtually flat 
($s \approx 0$) across the inner near-field mesh. Consequently, far-field freestream parameter updates are prevented from unphysically
disturbing the delicate boundary layer dynamics. For most viscous aerodynamic applications, a choice of $p \in [0.5, 1.0]$ 
provides an optimal balance between near-field boundary layer shielding and smooth spatial transition in the far field.

Once the discrete scalar field $\mathbf{s} \in \mathbb{R}^{N_c}$ is obtained by solving the linear system resulting from the spatial 
discretization of~\eqref{eq:laplace_shift_framework}, the high-dimensional conservative shift vector $\bar{\mathbf{x}}
(\boldsymbol{\mu}) \in \mathbb{R}^{n_{\mathrm{vars}} N_c}$ is constructed by mapping $\mathbf{s}_i$ to the local conservative variables
at each cell $i = 1, \dots, N_c$ according to:
\begin{equation} 
	\bar{\mathbf{x}}_i(\boldsymbol{\mu}) = 
	\begin{bmatrix} 
		\bar{\rho}_i(\boldsymbol{\mu}) \\[1ex]
		\overline{(\rho u)}_i(\boldsymbol{\mu}) \\[1ex] 
		\overline{(\rho v)}_i(\boldsymbol{\mu}) \\[1ex]
		\overline{(\rho w)}_i(\boldsymbol{\mu}) \\[1ex]
		\overline{(\rho e_t)}_i(\boldsymbol{\mu}) 
	\end{bmatrix} = 
	\begin{bmatrix} 
		\rho_{\infty}(\boldsymbol{\mu}) \\[1ex]
		(\rho u)_{\infty}(\boldsymbol{\mu}) \cdot \mathbf{s}_i \\[1ex]
		(\rho v)_{\infty}(\boldsymbol{\mu}) \cdot \mathbf{s}_i \\[1ex] 
		(\rho w)_{\infty}(\boldsymbol{\mu}) \cdot \mathbf{s}_i \\[1ex] 
		\left( (\rho e_t)_{\infty}(\boldsymbol{\mu}) - (\rho e)_{\infty}(\boldsymbol{\mu}) \right) \cdot \mathbf{s}_i + 
		(\rho e)_{\infty}(\boldsymbol{\mu})
	\end{bmatrix}
	\label{eq:shift_vector_construction}
\end{equation}
where $\rho_{\infty}(\boldsymbol{\mu})$, $(\rho u)_{\infty}(\boldsymbol{\mu})$, $(\rho v)_{\infty}(\boldsymbol{\mu})$, 
$(\rho w)_{\infty}(\boldsymbol{\mu})$, $(\rho e_t)_{\infty}(\boldsymbol{\mu})$, and $(\rho e)_{\infty}(\boldsymbol{\mu})$ denote the 
freestream conservative and internal energy quantities dictated by the parameter query $\boldsymbol{\mu}$. By constructing 
$\bar{\mathbf{x}}(\boldsymbol{\mu})$ via equation~\eqref{eq:shift_vector_construction}, boundary parameter variations are smoothly 
diffused into the interior volume mesh, maintaining complete spatial continuity, enforcing parametric boundary conditions 
\textit{a priori}, and entirely neutralizing Type~3 numerical snapping under the non-descriptor formulation.

\begin{remark}[Enforcement of parameter-dependent boundary conditions] 
	\label{rem:boundary_enforcement} 
	Unlike standard reduced-order trial spaces $\operatorname{Range}(\mathbf{V})$, which struggle to satisfy unseen boundary 
	constraints without resorting to penalty formulations that require empirical parameter tuning~\cite{cosimo2016general} or 
	\textit{a posteriori} constraint schemes (such as Lagrange multipliers) that inflate online computational 
	complexity~\cite{kalashnikova2014penalty}, the spatial shift field $\bar{\mathbf{x}}(\boldsymbol{\mu})$ is engineered to 
	satisfy parameter-dependent boundary conditions (e.g., far-field freestream updates or wall surface deformations) 
	\textit{a priori} by construction. Enforcing these boundary conditions directly through $\bar{\mathbf{x}} (\boldsymbol{\mu})$ 
	eliminates steep residual gradients at boundary cells, effectively preventing the grid-resolution scaling artifacts that 
	trigger numerical snapping in non-descriptor formulations.
\end{remark}

\begin{remark}[Scope of wall boundary enforcement: viscous vs. inviscid flows]
	\label{rem:viscous_vs_inviscid_bc}
	The field-theoretic shift vector construction~\eqref{eq:shift_vector_construction} addresses viscous no-slip wall boundary 
	conditions ($\mathbf{u}_{\mathrm{wall}} = \mathbf{0}$), where geometric parameterizations deform mesh coordinates while maintaining a homogeneous velocity state on solid surfaces. In contrast, inviscid slip boundary conditions ($\mathbf{u} \cdot \mathbf{n}(\boldsymbol{\mu}) = 0$) involve a parameter-dependent surface normal field $\mathbf{n}(\boldsymbol{\mu})$ on deforming geometries, requiring the shift vector field to satisfy a dynamic directional constraint rather than a scalar Dirichlet-type extension. Extending the state-shifting framework to guarantee \textit{a priori} enforcement of parameter-dependent inviscid normal constraints presents distinct differential mechanics challenges and is deferred to future work.
\end{remark}

\begin{remark}[Target of data compression] 
	It is worth noting that once the HDM of interest is exercised to compute the snapshots, the reference state $\mathbf{x}_{\mathrm{ref}}$ is set, and the shift vector $\bar{\mathbf{x}}(\boldsymbol{\mu})$ is constructed, data compression is not applied directly to the raw solution snapshots $\mathbf{x}$. Instead, data compression is performed on the shifted snapshot states, namely $\mathbf{x} - \mathbf{x}_{\mathrm{ref}} - \bar{\mathbf{x}}(\boldsymbol{\mu})$. This ensures that the extracted right ROB vectors satisfy homogeneous boundary conditions, thereby preserving exact satisfaction of parameter-dependent boundary conditions in the reconstructed state regardless of the reduced representation.
\end{remark}

\section{Applications}
\label{sec:applications}

To evaluate the theoretical insights and mechanistic conclusions established in the preceding sections, we assess the performance of the proposed unified state-shifting framework across a series of parametric numerical benchmarks. The primary objective of this section is to illustrate the distinct operational signatures of Type~1 (intrinsic) and Type~3 (numerical) snapping, while providing empirical validation for the dual-action capabilities of our framework: guaranteeing exact \textit{a priori} boundary condition enforcement and restoring continuous parametric sensitivity without resorting to expensive online penalty terms or \textit{a posteriori} constraint schemes.

Regarding Type~1 snapping, attention is focused on the instructive edge case arising from Proposition~\ref{pro:p4}, where the projected forcing constraint $\mathbf{W}^T \Delta \mathbf{b} = \mathbf{0}$ acts as a necessary and sufficient condition for generalized coordinate stagnation. This intrinsic mechanism and its algebraic shift-vector mitigation are demonstrated using a linear static structural analysis benchmark.

The numerical investigations targeting Type~3 snapping center on two aerodynamic fluid flow applications designed to balance baseline accessibility with industrial complexity:
\begin{enumerate} 
	\item An academic-scale, two-dimensional parametric NACA airfoil configuration, serving as a transparent and easily reproducible testbed for evaluating fundamental snapping mechanisms and exact \textit{a priori} boundary transport across diverse parameter combinations.  
	\item An industrial-grade, three-dimensional aircraft configuration based on the publicly available NASA Common Research Model (CRM)~\cite{vassberg2008development}, demonstrating the scalability, robustness, and practical efficacy of the proposed framework within large-scale computational pipelines.
\end{enumerate}

To ensure a rigorous and direct comparison across both aerodynamic benchmarks, the underlying PROMs are constructed using an un-hyperreduced, affine approximation architecture solved via the LSPG projection method~\cite{carlberg2011efficient}. For both the NACA airfoil and NASA CRM configurations, online performance is systematically evaluated across three representative formulation variants: (i) non-descriptor without a shift, (ii) descriptor without a shift, and (iii) non-descriptor augmented with the spatially-varying Laplace-governed shift vector.

In both fluid applications, the high-fidelity snapshot solutions are obtained by solving the full, steady compressible nonlinear Navier-Stokes equations, coupled with the Spalart-Allmaras one-equation turbulence model. Spatial discretization relies on a second-order, MUSCL-type upwind cell reconstruction scheme for numerical flux evaluation. Parameter spaces are sampled using a structured Sobol sequence to generate the offline training sets. Furthermore, because these benchmarks operate exclusively in the steady-state regime, the POD bases are extracted without energy truncation, fixing the basis dimension to match the total number of collected snapshot states and ensuring full snapshot energy retention.

All high-fidelity simulations and PMOR workflows evaluated in this section are implemented within the compressible finite volume flow solver AERO-F, whose computational algorithms and academic implementation are outlined in~\cite{farhat2003application} and available at~\cite{frg_aero_bitbucket}. All computations are executed in double-precision arithmetic, with flow simulations formulated in non-dimensional variables. Because AERO-F operates natively as a three-dimensional spatial solver, the academic airfoil benchmark is simulated in a quasi-two-dimensional configuration by extruding the planar geometry into a single element layer with symmetric sidewall boundary conditions.

\subsection{Linear static analysis of a structural system}
\label{subsec:structuralsystem}

To empirically validate the algebraic mitigation of Type~1 (intrinsic) snapping derived in Section~\ref{subsec:mitigation_linear}, we consider a parametric linear static structural system governing the transverse deflection $w(x;\boldsymbol{\mu})$ of a straight cantilevered beam aligned with the $x$-axis. The beam is subject to a parameter-dependent distributed transverse load $q(x;\boldsymbol{\mu})$, where $\boldsymbol{\mu} \in \mathcal{D}$ controls the spatial variation of the load. The beam possesses a total length $L$, constant flexural rigidity $EI$, and uniform cross-sectional properties. For the numerical example, $L = 1$ and $EI = 1$.

The physical domain is discretized into $N_e = 1\,000$ Euler--Bernoulli beam elements with two degrees of freedom (DoFs) per node -- one vertical translational and one rotational -- yielding a high-dimensional system size of $N = 2 N_e = 2\,000$ DoFs. The governing physical system and its discretization are schematically illustrated in Figure~\ref{fig:cantilever_beam_setup}. This spatial 
discretization yields the high-dimensional linear algebraic system
\begin{equation*}
	\mathbf{A} \mathbf{w}(\boldsymbol{\mu}) = \mathbf{b}(\boldsymbol{\mu})
\end{equation*}
where $\mathbf{A} \in \mathbb{R}^{N \times N}$ is the parameter-independent FE stiffness matrix, $\mathbf{w}(\boldsymbol{\mu}) \in \mathbb{R}^N$ is the state vector of nodal vertical displacements and rotations, and $\mathbf{b}(\boldsymbol{\mu}) \in \mathbb{R}^N$ is the parameter-dependent load vector resulting from the discretization of the applied load distribution $q(x;\boldsymbol{\mu})$. For simplicity, the parameter vector is defined directly as the nodal load intensity values, $\boldsymbol{\mu} = [q_0 \quad q_1 \quad q_2 \quad \dots \quad q_{N_e}]^T$, assuming a piecewise linear load distribution between adjacent nodes.

\begin{figure}[pos=h!]
    \centering
    \includegraphics[width=0.60\linewidth]{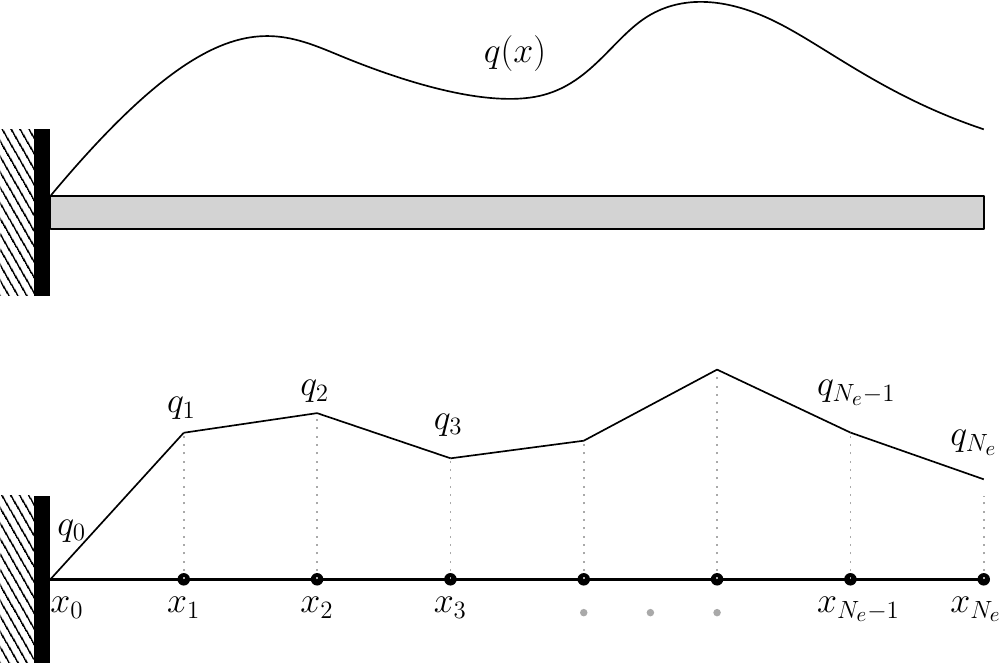}
    \caption{Schematic of a cantilevered beam subject to a parameter-dependent distributed load $q(x)$ (top), and its corresponding 
	FE discretization into $N_e$ Euler--Bernoulli beam elements (bottom).}
    \label{fig:cantilever_beam_setup}
\end{figure}

Offline data generation proceeds by collecting $n_s = 10$ high-dimensional state snapshots $\mathbf{w}(\boldsymbol{\mu}_j)$ corresponding to Gaussian-shaped distributed loads centered at randomized spatial positions across the beam span, as depicted in Figure~\ref{fig:beam_loads}. Performing an untruncated POD on these snapshots yields a right ROB $\mathbf{V} \in \mathbb{R}^{N \times n}$ of dimension $n = 10$. Employing a standard Galerkin projection ($\mathbf{W} = \mathbf{V}$) yields the baseline linear PROM:
\begin{equation*}
	(\mathbf{V}^T \mathbf{A} \mathbf{V}) \mathbf{y}(\boldsymbol{\mu}) = \mathbf{V}^T \mathbf{b}(\boldsymbol{\mu})
\end{equation*}
where $\mathbf{y}(\boldsymbol{\mu}) \in \mathbb{R}^n$ represents the online vector of generalized coordinates.

\begin{figure}[pos=h!]
    \centering
    \includegraphics[width=0.60\linewidth]{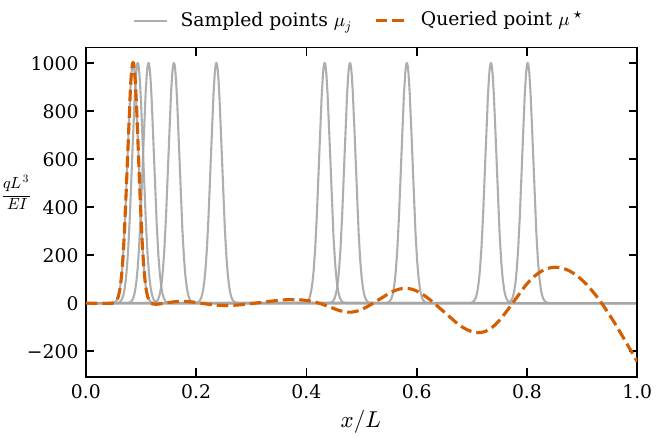}
    \caption{Normalized spatial load distributions ${q L^3}/{EI}$ on the cantilevered beam: offline training configurations 
	$\{\boldsymbol{\mu}_j\}_{j=1}^{10}$ used to construct the right ROB $\mathbf{V}$ (solid gray lines), and the un-sampled 
	online parameter query $\boldsymbol{\mu}^\star$ (dashed orange line).}
    \label{fig:beam_loads}
\end{figure}

To test the baseline PROM, the online parameter query $\boldsymbol{\mu}^\star$ is constructed such that its multi-peak load profile $q(x;\boldsymbol{\mu}^\star)$ (Figure~\ref{fig:beam_loads}) yields a forcing vector $\mathbf{b}(\boldsymbol{\mu}^\star)$ that intentionally satisfies the exact intrinsic snapping condition derived in Proposition~\ref{pro:p4} for Galerkin projections:
\begin{equation*}
	\mathbf{V}^T \big( \mathbf{b}(\boldsymbol{\mu}^\star) - \mathbf{b}(\boldsymbol{\mu}_1) \big) = \mathbf{0}
\end{equation*}
where $\boldsymbol{\mu}_1$ corresponds to a specific offline training parameter point.

As demonstrated in Figure~\ref{fig:beam_deflection}, the reconstructed deflection profile from the unshifted baseline PROM evaluated at $\boldsymbol{\mu}^\star$ fails to track the true HDM solution at $\boldsymbol{\mu}^\star$. Instead, the unshifted PROM solution locks identically onto the offline HDM snapshot computed at $\boldsymbol{\mu}_1$, exhibiting the classical snapshot latching signature of Type~1 intrinsic snapping. Conversely, when the online system is augmented using the parameter-dependent algebraic shift vector (Equation~\eqref{eq:forcing_vector}), the test-space orthogonality condition is successfully broken. The shifted PROM completely eliminates snapshot latching and achieves excellent agreement with the high-fidelity HDM solution across the entire length of the beam.

\begin{figure}[pos=h!]
    \centering
    \includegraphics[width=0.60\linewidth]{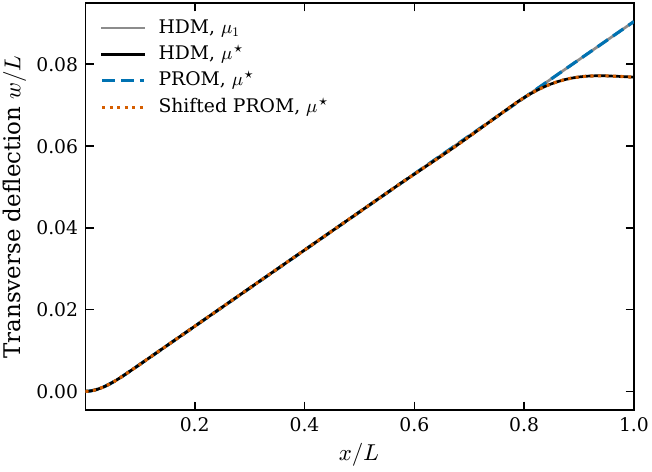}
    \caption{Normalized transverse beam deflection $w/L$ along the normalized span $x/L$: comparison of the unshifted baseline PROM 
	solution and the state-shifted PROM solution evaluated at the query parameter $\boldsymbol{\mu}^\star$ against the true HDM 
	solution at $\boldsymbol{\mu}^\star$ and the training snapshot at $\boldsymbol{\mu}_1$.}
    \label{fig:beam_deflection}
\end{figure}

\subsection{Steady viscous flow around a NACA airfoil under multi-boundary parameterizations}
\label{subsec:naca}

To systematically investigate the mechanisms governing Type~3 numerical snapping across distinct operational regimes, we consider steady viscous flow around a parametric NACA airfoil configuration with chord length $c=1$. The computational domain is discretized using an O-grid mesh extending to an outer boundary radius of approximately $83$~chord lengths. The mesh consists of $812\,098$~nodes and $4\,052\,396$~tetrahedral elements, with near-wall grid refinement engineered to enforce $y^+ \approx 1$ across the baseline freestream Mach number range $M_\infty \in [0.4, 0.6]$. Far-field boundary conditions are enforced using Steger--Warming flux-vector splitting, while solid surfaces are subjected to viscous no-slip, adiabatic wall conditions.

To evaluate the interaction between equation scaling and boundary parameterizations across distinct sources of parametric variation, online performance is assessed across three setups: (i)~parametric far-field conditions only, (ii)~parametric wall geometry only, and (iii)~combined parametric far-field and wall geometry. In each test case, three LSPG PROM configurations are compared: (a)~unshifted non-descriptor, (b)~unshifted descriptor, and (c)~Laplace-shifted non-descriptor.

\subsubsection{Parametric far-field conditions only}
\label{subsubsec:naca_farfield}

We first consider a two-dimensional parameter domain $\boldsymbol{\mu} = [M_\infty \quad \alpha]^T \in \mathcal{D}_{\text{ff}}$, governing variations in freestream Mach number $M_\infty \in [0.4, 0.6]$ and angle of attack $\alpha \in [-5^\circ, 5^\circ]$. The airfoil geometry remains fixed to the baseline symmetric NACA0012 profile. Sobol sampling of $\mathcal{D}_\text{ff}$ at 40 parameter points (see Figure~\ref{fig:naca-ff-sobol}), snapshot collection at these points, and snapshot compression yield a right ROB of dimension $n = 40$.

\begin{figure}[pos=h!] 
	\centering 
	\includegraphics[width=0.60\linewidth]{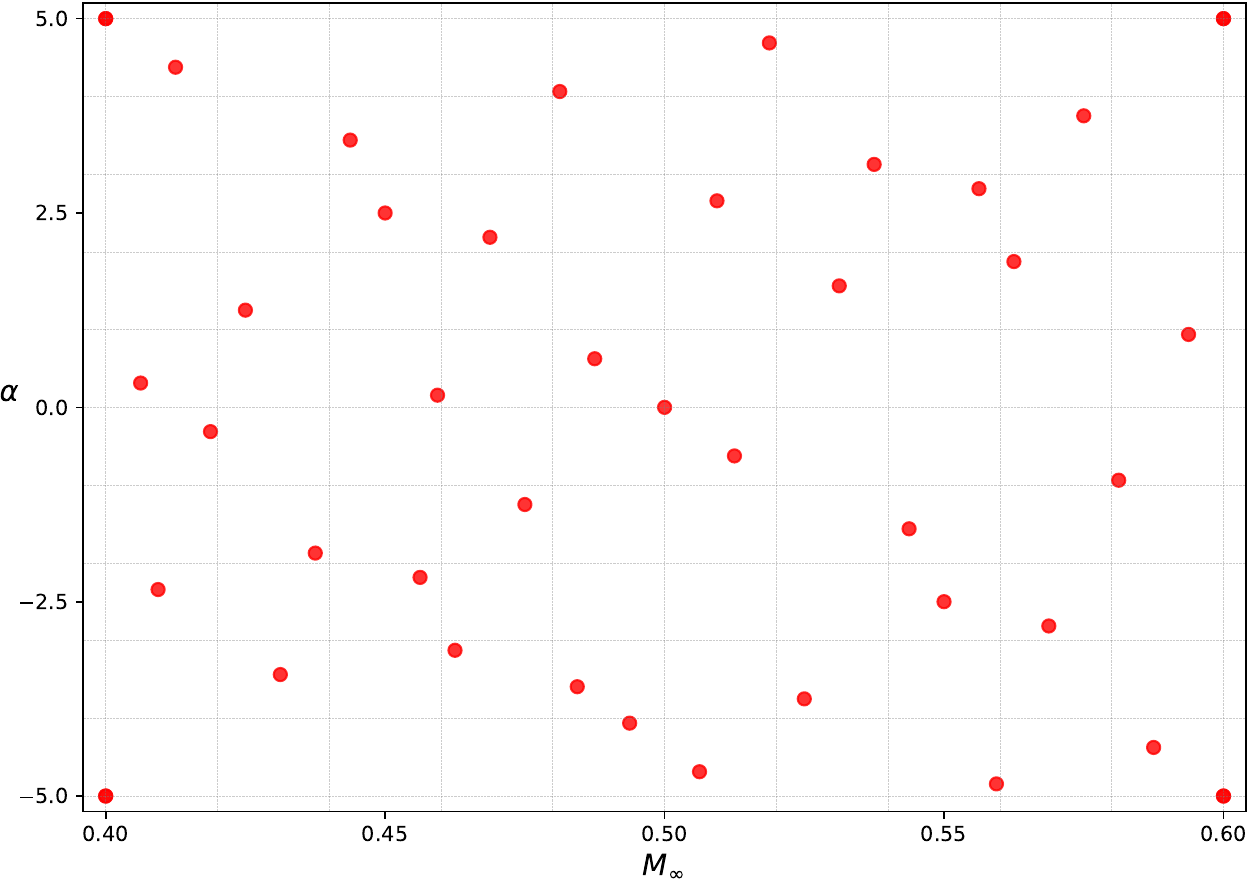} 
	\caption{40 HDM training points generated by a Sobol sequence in $\mathcal{D}_\text{ff}$.} 
	\label{fig:naca-ff-sobol}
\end{figure}

To highlight the spatial scaling mechanics detailed in Section~\ref{sec:numerical_snapping_mechanism}, we evaluate the online 
performance of the three model configurations: the unshifted non-descriptor PROM, the unshifted descriptor PROM, and the 
Laplace-shifted non-descriptor PROM. For each model configuration under consideration, we construct an LSPG PROM of dimension $n=40$ 
to predict the flow past the NACA0012 profile at the unseen parameter point $\boldsymbol{\mu}^\star = [0.54 \quad 5^\circ]^T$. 
Figure~\ref{fig:naca_farfield_results} compares the pressure coefficient profiles predicted at $\boldsymbol{\mu}^\star$ by these PROMs
against the corresponding HDM predictions at both $\boldsymbol{\mu}^\star$ and the training point $\boldsymbol{\mu}^\dagger = [0.5375 
\quad 3.125^\circ]^T$.
\begin{itemize} 
	\item The non-descriptor PROM without a shift snaps onto the training snapshot computed at $\boldsymbol{\mu}^\dagger$ 
		(Figure~\ref{fig:naca_farfield_results}(a)).  
	\item The descriptor PROM without a shift reproduces accurately the pressure coefficient profile predicted by the HDM 
		(Figure~\ref{fig:naca_farfield_results}(b)).  
	\item The non-descriptor PROM with a Laplace shift also reproduces accurately the pressure coefficient profile predicted by 
		the HDM (Figure~\ref{fig:naca_farfield_results}(c)).
\end{itemize}

\begin{figure}[pos=h!] 
	\centering 
	\includegraphics[width=1\textwidth]{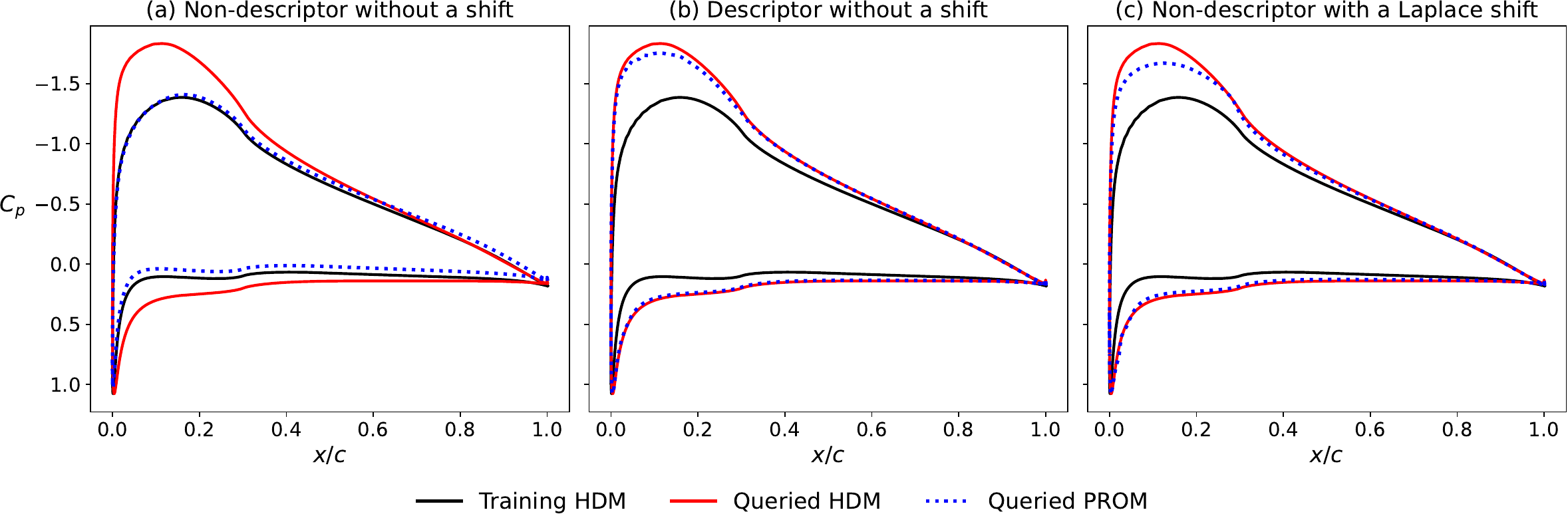} 
	\caption{Training HDM at $(M_\infty, \alpha) = (0.5375, 3.125^\circ)$, PROM queried at $(M_\infty, \alpha) = (0.54, 
	5.0^\circ)$.
	Online performance comparison for the parametric far-field case: (a) non-descriptor without a shift, (b) descriptor without a shift, and (c) non-descriptor with a Laplace-governed shift.} 
	\label{fig:naca_farfield_results}
\end{figure}

Taken together, these results demonstrate that both the unshifted descriptor formulation and the Laplace-shifted non-descriptor 
formulation successfully mitigate the rigid snapshot latching characteristic of Type~3 numerical snapping in this far-field 
parameterization. However, unlike the descriptor formulation -- which inherently compromises near-field error minimization in highly 
refined boundary layer cells -- the Laplace-shifted non-descriptor framework enforces parameter-dependent far-field boundary conditions
\textit{a priori} and restores continuous parametric regression while preserving the exact numerical scaling properties required for 
near-field fidelity.

\subsubsection{Parametric wall geometry only}
\label{subsubsec:naca_geometry}

Next, we isolate parameterizations acting exclusively on the solid wall boundary. The freestream operating conditions are held 
constant at $(M_\infty, \alpha) = (0.6, 3^\circ)$, while the airfoil geometry varies over a three-dimensional parameter domain 
$\boldsymbol{\mu} = [m \quad p \quad t]^T \in \mathcal{D}_{\text{geom}}$. In this setup, $m$ represents the maximum camber, $p$ 
denotes its location, and $t$ dictates the maximum thickness. A 10-point Sobol sampling of $\mathcal{D}_{\text{geom}}$ 
(Figure~\ref{fig:naca-geom-sobol}) suffices to construct a sufficiently accurate right ROB.

\begin{figure}[pos=h!] 
	\centering 
	\includegraphics[width=0.60\linewidth]{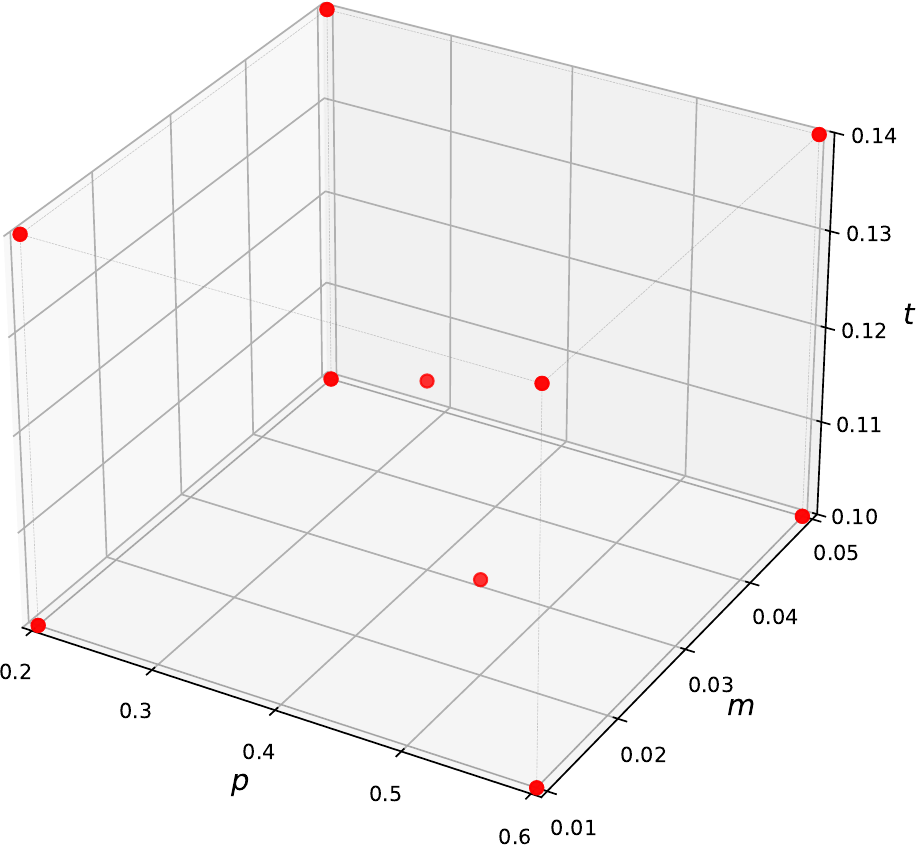} 
	\caption{10 HDM training points generated by a Sobol sequence in $\mathcal{D}_\text{geom}$.} 
	\label{fig:naca-geom-sobol}
\end{figure}

To evaluate the inverse spatial scaling dynamics under wall geometry variations, we assess the same three PROM configurations: 
unshifted descriptor, unshifted non-descriptor, and Laplace-shifted non-descriptor. For each model configuration, we construct an 
LSPG PROM of dimension $n=10$ to predict the flow past the airfoil geometry defined by the unseen parameter point 
$\boldsymbol{\mu}^\star = [0.03 \quad 0.6 \quad 0.10]^T$. Figure~\ref{fig:naca_geometry_results} compares the pressure coefficient 
profiles predicted at $\boldsymbol{\mu}^\star$ by these PROMs against the corresponding HDM predictions at both 
$\boldsymbol{\mu}^\star$ and the training point $\boldsymbol{\mu}^\dagger = [0.02 \quad 0.5 \quad 0.11]^T$:
\begin{itemize} 
	\item The non-descriptor PROM without a shift accurately reproduces the pressure coefficient profile predicted by the HDM 
		(Figure~\ref{fig:naca_geometry_results}(a)).  
	\item The descriptor PROM without a shift snaps onto the training snapshot computed at $\boldsymbol{\mu}^\dagger$ 
		(Figure~\ref{fig:naca_geometry_results}(b)).  
	\item The non-descriptor PROM with a Laplace shift reproduces the pressure coefficient profile predicted by the HDM even more 
		accurately than the unshifted variant (Figure~\ref{fig:naca_geometry_results}(c)).
\end{itemize}

\begin{figure}[pos=h!] 
	\centering 
	\includegraphics[width=1\textwidth]{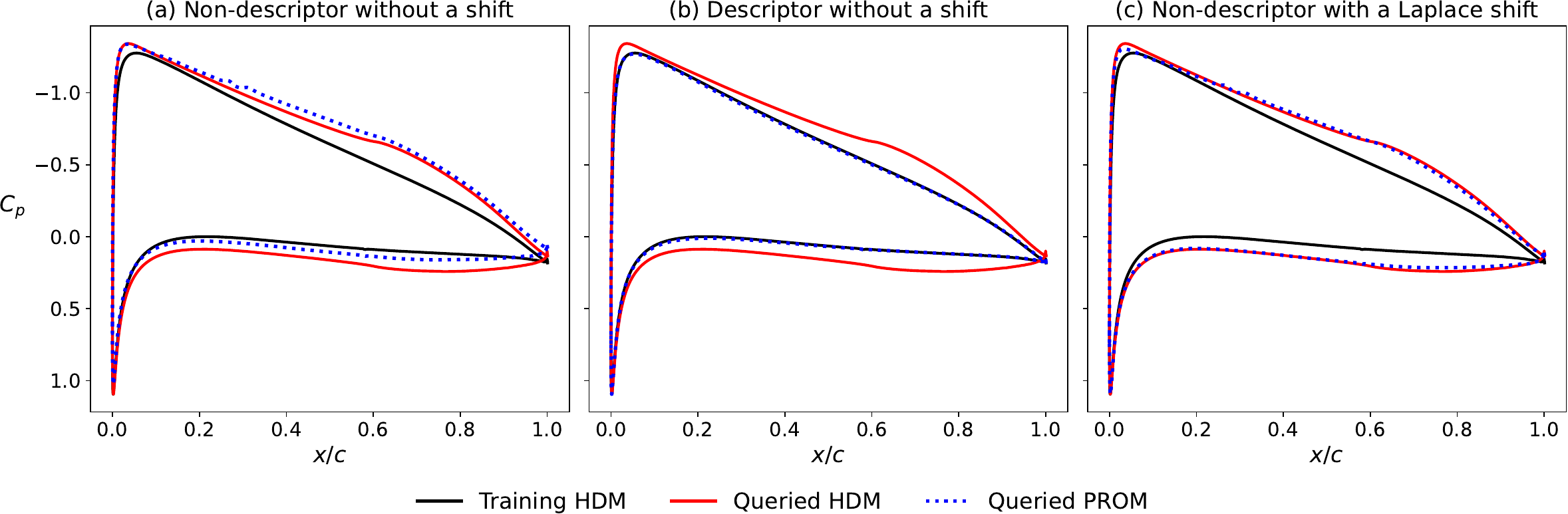} 
	\caption{Training HDM at $(m, p, t) = (0.02, 0.5, 0.11)$, PROM queried at $(m, p, t) = (0.03, 0.6, 0.10)$.
	Online performance comparison for the parametric wall geometry case: (a) non-descriptor without a shift, (b) descriptor 
	without a shift, and (c) non-descriptor with a Laplace-governed shift.}
	\label{fig:naca_geometry_results}
\end{figure}

In summary, these observations confirm that both the unshifted non-descriptor model and its Laplace-augmented counterpart robustly 
suppress the rigid snapshot latching characteristic of Type~3 numerical snapping under wall geometry variations. By contrast, the 
descriptor formulation succumbs to snapping in this setting because its unweighted residual allows the massive cell volumes of the 
far field to mask near-field surface updates. Importantly, embedding the Laplace-governed shift vector into the non-descriptor 
baseline introduces no adverse artifacts or modeling bias, proving that the proposed state-shifting strategy operates benignly when 
viscous wall boundary conditions are homogeneous.

\subsubsection{Combined parametric far-field and wall geometry}
\label{subsubsec:naca_combined}

Finally, we evaluate the coupled multi-boundary scenario where parameter variations simultaneously govern far-field conditions and 
surface geometry updates. We define a three-dimensional parameter domain $\boldsymbol{\mu} = [M_\infty \quad \alpha \quad p]^T \in 
\mathcal{D}_{\text{comb}}$, varying freestream Mach number $M_\infty \in [0.4, 0.6]$, angle of attack $\alpha \in [-5^\circ, 5^\circ]$,
and the location of maximum camber $p \in [0.3, 0.6]$, with fixed maximum camber $m = 0.03$ and thickness $t = 0.12$. Offline 
snapshot generation yields an untruncated reduced basis dimension of $n = 110$ (see Figure~\ref{fig:naca-combined-sobol}).

\begin{figure}[pos=h!] 
	\centering 
	\includegraphics[width=0.6\textwidth]{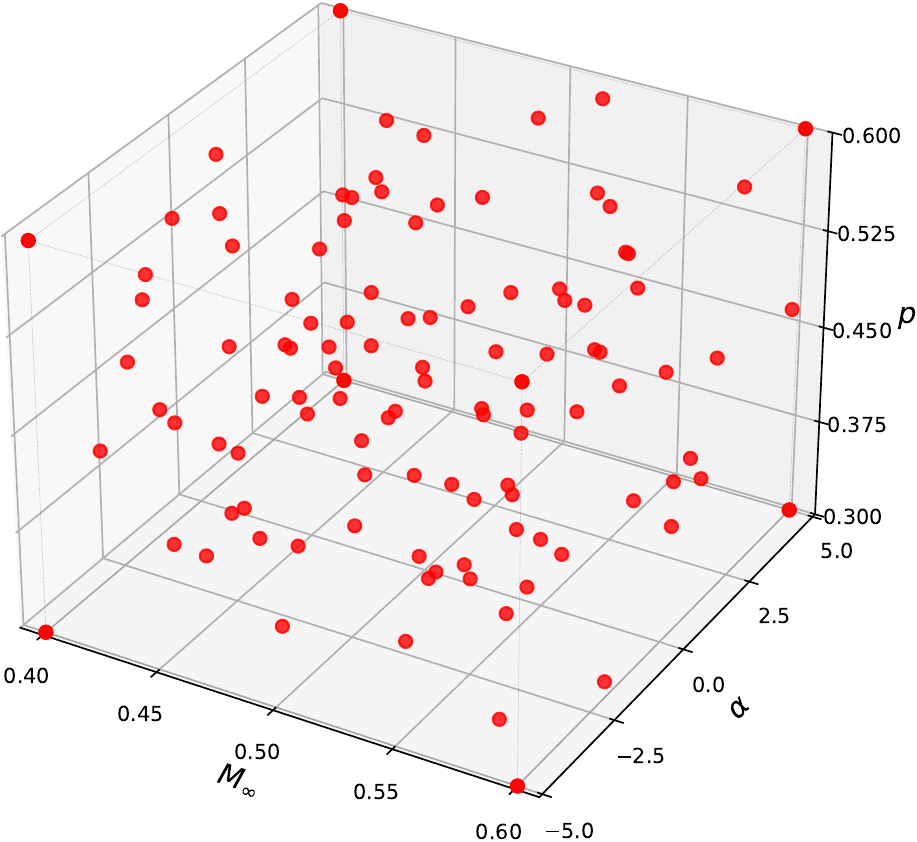} 
	\caption{110 HDM training points generated by a Sobol sequence in $\mathcal{D}_\text{comb}$.} 
	\label{fig:naca-combined-sobol}
\end{figure}

To evaluate the coupled multi-boundary scenario, online predictions across the three PROM variants are assessed to test their 
susceptibility to historical formulation deadlock. Specifically, each dimension-$n=110$ LSPG PROM is queried at the unseen parameter 
point $\boldsymbol{\mu}^\star = [0.6 \quad -2.5^\circ \quad 0.3]^T$ to predict the flow field across simultaneous geometric and 
freestream variations. Figure~\ref{fig:naca_combined_results} compares the pressure coefficient profiles predicted at 
$\boldsymbol{\mu}^\star$ by these PROMs against the corresponding HDM predictions at both $\boldsymbol{\mu}^\star$ and the training 
point $\boldsymbol{\mu}^\dagger = [0.596875 \quad -1.71875^\circ \quad 0.304688]^T$. Without a shift, the non-descriptor PROM snaps 
onto the training snapshot at $\boldsymbol{\mu}^\dagger$ (Figure~\ref{fig:naca_combined_results}(a)). While the unshifted descriptor 
model accurately captures the high-fidelity target profile (Figure~\ref{fig:naca_combined_results}(b)), the Laplace-shifted 
non-descriptor PROM reproduces the HDM pressure distribution with even greater precision (Figure~\ref{fig:naca_combined_results}(c)).

\begin{figure}[pos=h!] 
	\centering 
	\includegraphics[width=1.0\textwidth]{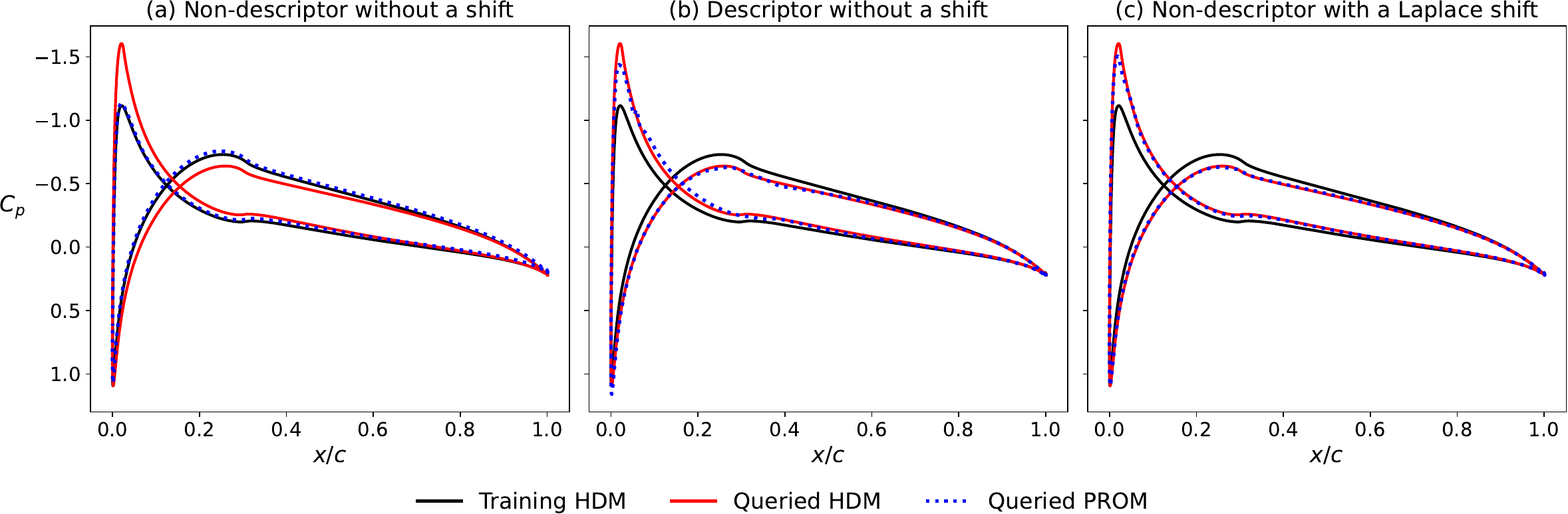} 
	\caption{Training HDM at $(M_\infty, \alpha, p) = (0.596875, -1.71875^\circ, 0.304688)$, PROM queried at $(M_\infty, \alpha, p)
	= (0.6, -2.5^\circ, 0.3)$. 
	Online performance comparison for the combined far-field and geometry parameterization: (a) non-descriptor without a shift, 
	(b) descriptor without a shift, and (c) non-descriptor with a Laplace-governed shift.} 
	\label{fig:naca_combined_results} 
\end{figure}

Crucially, this combined benchmark demonstrates the overarching advantage of our unified state-shifting framework. By universally 
adhering to the non-descriptor formulation and absorbing parameter-dependent boundary updates \textit{a priori} via the 
field-theoretic Laplace extension field $\bar{\mathbf{x}}(\boldsymbol{\mu})$, the framework resolves the historical structural 
deadlock without requiring \textit{a posteriori} constraint terms or compromising near-field simulation accuracy.

\subsection{Steady viscous flow around a NASA CRM under parametric far-field conditions}
\label{subsec:crm}

To evaluate the scalability and practical efficacy of the proposed state-shifting framework within a large-scale computational 
pipeline, we consider steady compressible viscous flow around the three-dimensional NASA CRM described 
in~\cite{vassberg2008development}. The outer boundary of the spherical computational fluid domain is placed at a radius of 
approximately $10\,\ell$, where $\ell$ denotes the reference length measured from the midpoint of the CRM's longitudinal extent to 
its tail. The physical domain is discretized into $9\,764\,468$ nodes and $57\,956\,418$ elements. To fully resolve boundary layer 
effects, the mesh is constructed to have an average $y^+\approx 0.77$ across all sampled parameter points.

The two-dimensional parameter domain is defined as $\boldsymbol{\mu} = [M_\infty \quad \alpha]^T \in \mathcal{D}_\text{ff}$, where the
freestream Mach number ranges over $M_\infty \in [0.75, 0.85]$ and the angle of attack varies within $\alpha \in [0^\circ, 3.5^\circ]$.
Offline sampling of $\mathcal{D}_\text{ff}$ is performed using a 20-point Sobol sequence to ensure uniform coverage across the 
parameter space. High-fidelity snapshots are obtained at each sample point by solving the steady compressible Navier--Stokes 
equations coupled with the Spalart--Allmaras turbulence model. The right ROB $\mathbf{V}$ is constructed via POD without energy 
truncation, retaining $n = n_s = 20$ snapshot states.

Under these parametric far-field conditions, online predictive performance is evaluated across the three representative formulation 
variants:
\begin{itemize}
    \item \textit{Unshifted non-descriptor formulation:} Subject to the $1/\Omega_i$ spatial scaling imbalance detailed in 
	  Section~\ref{sec:nonlinear}, wherein the unperturbed near-field residual dominates the global $l_2$-norm, driving rigid 
	  snapshot latching instead of smooth parametric tracking.
    \item \textit{Unshifted descriptor formulation:} Retains the unweighted residual to track freestream variations smoothly, 
	    though it may degrade near-field boundary layer fidelity targets under standard solver tolerances.
    \item \textit{Laplace-shifted non-descriptor formulation:} Embeds the field-theoretic shift vector 
	  $\overline{\mathbf{x}}(\boldsymbol{\mu})$ governed 
	  by~\eqref{eq:laplace_shift_framework}--\eqref{eq:shift_vector_construction} directly into the non-descriptor trial space. 
	  This enforces parameter-dependent far-field boundary conditions \textit{a priori}, which in principle eliminates Type~3 
	  numerical snapping while preserving near-field simulation accuracy.
\end{itemize}

\begin{figure}[pos=h!]
\centering
\includegraphics[width=0.6\linewidth]{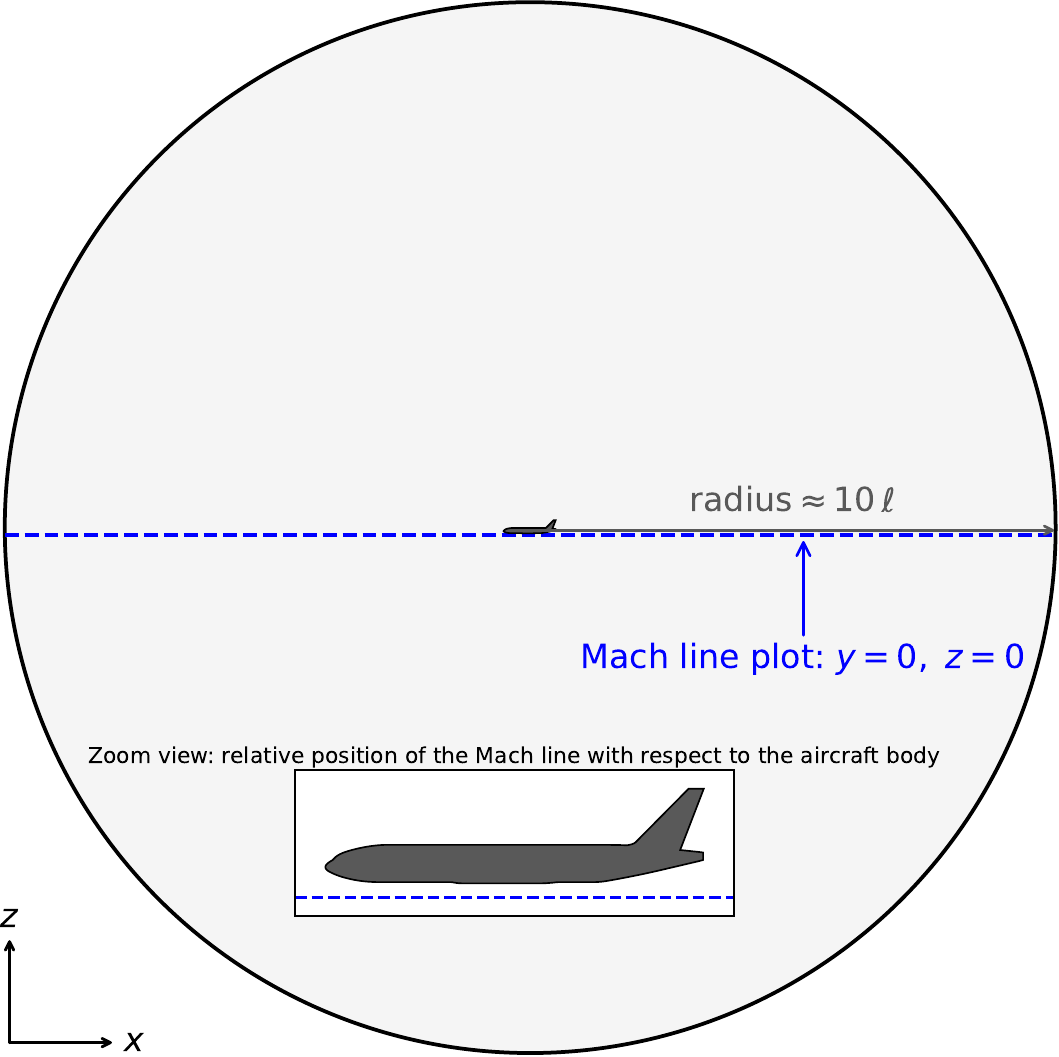}
\caption{Symmetry plane of the spherical computational fluid domain for the NASA CRM. The Mach distribution shown in 
	Figure~\ref{fig:crm_mach_dist} is plotted along the blue dashed line.}
\label{fig:crm_domain}
\end{figure}

Figure~\ref{fig:crm_domain} illustrates the symmetry plane of the computational domain and identifies the line along which the
far-field Mach-number distribution is evaluated. The 20 offline training points generated by the Sobol sequence are shown in
Figure~\ref{fig:crm_sobol}. For the online assessment, the three PROM variants are queried at
$\boldsymbol{\mu}^\star = [0.83 \quad 1.0^\circ]^T$, an unsampled point within $\mathcal{D}_{\mathrm{ff}}$.


\begin{figure}[pos=h!]
\centering
\includegraphics[width=0.6\linewidth]{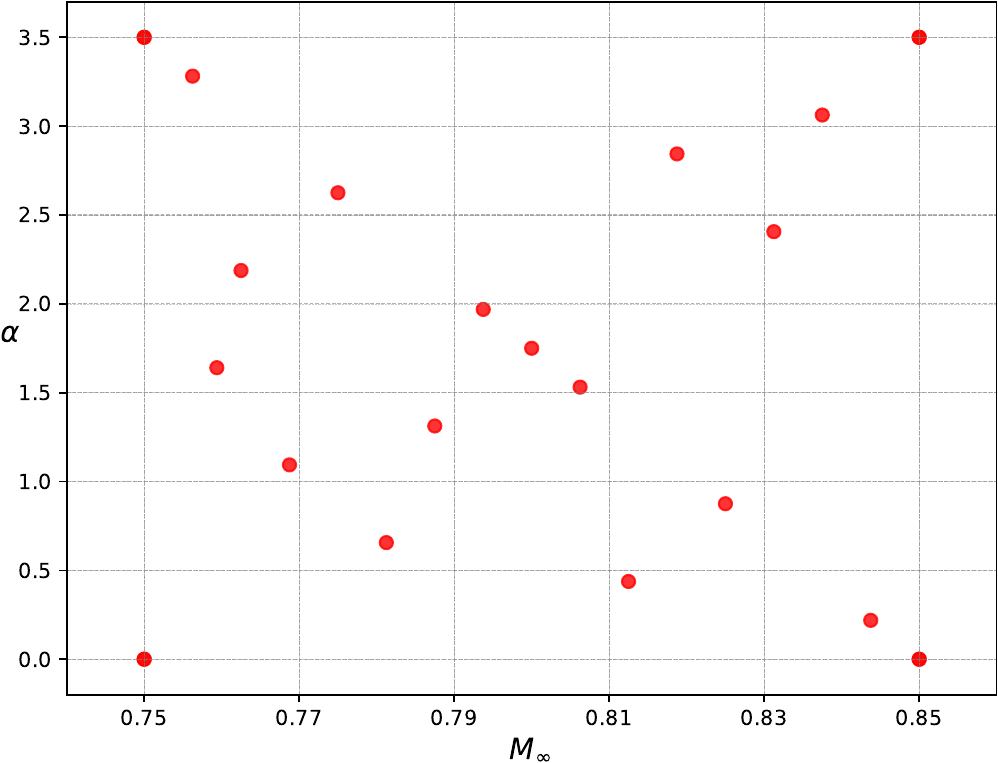}
\caption{20 HDM training points generated by the Sobol sequence in $\mathcal{D}_\text{ff}$.} 
	\label{fig:crm_sobol}
\end{figure}

Unlike the NACA airfoil cases, no Type~3 snapping is observed for any of the three PROM configurations in this NASA CRM benchmark. 
This observation is not inconsistent with the theoretical developments presented above. The analysis in 
Section~\ref{sec:numerical_snapping_mechanism} identifies structural and numerical conditions that promote snapping, but does
not establish that their presence is sufficient to guarantee its occurrence in every configuration. In contrast, the mitigation
strategy based on the state-dependent shift was constructed specifically to eliminate the identified mechanism by enforcing the
parameter-dependent boundary conditions \textit{a priori}. Thus, the absence of snapping in this case does not weaken the proposed
mechanism; rather, it illustrates that the sufficient conditions for snapping need not be activated for every parameter query or
every problem instance.

Although no snapshot latching is observed, the expected effect of the residual scaling on boundary-condition satisfaction is
clearly visible. Figure~\ref{fig:crm_mach_dist} compares the Mach-number distributions obtained along the line shown in
Figure~\ref{fig:crm_domain}. The unshifted non-descriptor PROM fails to recover the prescribed far-field Mach number
$M_\infty=0.83$, as shown in Figure~\ref{fig:crm_mach_dist}(a). This behavior is consistent with the analysis of
Section~\ref{sec:numerical_snapping_mechanism}: the inverse cell-volume scaling strongly emphasizes the finely discretized
near-field region, while the far-field condition receives comparatively little weight in the residual minimization. In contrast,
both the unshifted descriptor PROM and the Laplace-shifted non-descriptor PROM satisfy the prescribed far-field Mach-number
condition, as shown in Figures~\ref{fig:crm_mach_dist}(b) and~\ref{fig:crm_mach_dist}(c), respectively.

\begin{figure}[pos=h!]
\centering
\includegraphics[width=1\linewidth]{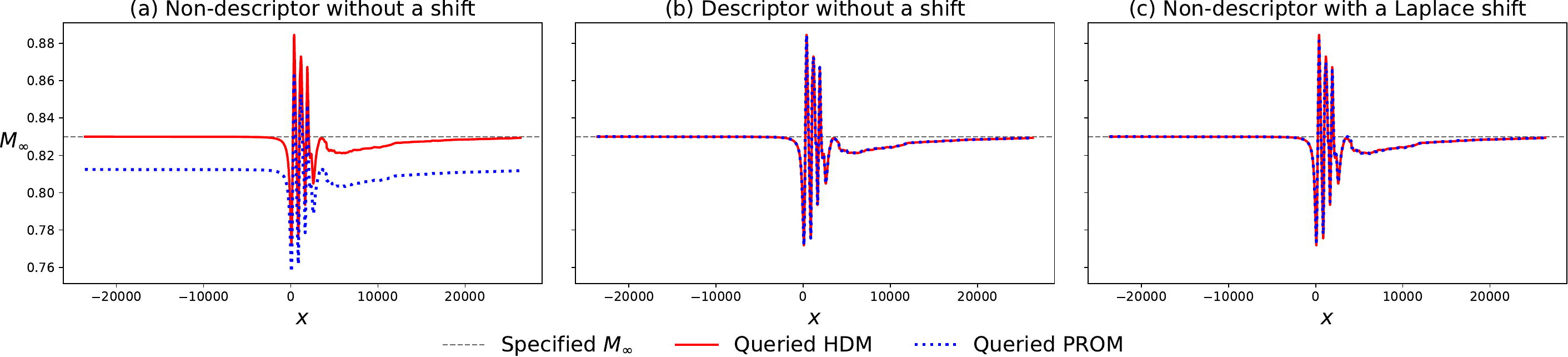}
\caption{Parameter point $(M_\infty, \alpha)=(0.83, 1.0^\circ)$ queried. \
Online performance comparison for the far-field parameterization: (a) non-descriptor without a shift, which violates the
far-field $M_\infty$ condition, (b) descriptor without a shift, and (c) non-descriptor with a Laplace-governed shift, both of
which respect the far-field condition.}
\label{fig:crm_mach_dist}
\end{figure}

\begin{figure}[pos=h!]
\centering
\includegraphics[width=1\linewidth]{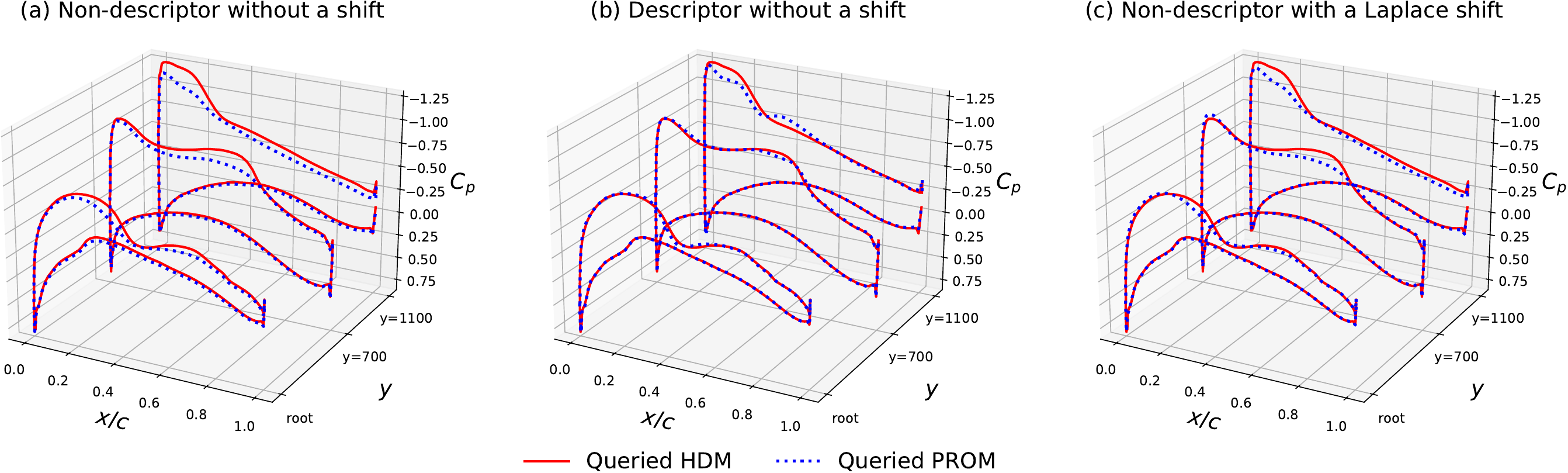}
\caption{PROM queried at $(M_\infty, \alpha)=(0.83, 1.0^\circ)$. \
Online performance comparison for the far-field parameterization: (a) non-descriptor without a shift, (b) descriptor
without a shift, and (c) non-descriptor with a Laplace-governed shift.}
\label{fig:crm_cp}
\end{figure}

Figure~\ref{fig:crm_cp} shows the surface pressure coefficient distributions as functions of $x/c$, where $c$ denotes the local chord 
length, at three different wing sections. All three PROMs provide a reasonable approximation of the HDM solution at the queried 
parameter point. Nevertheless, the pressure coefficient predicted by the unshifted non-descriptor PROM is slightly less accurate than 
those obtained with the descriptor and Laplace-shifted non-descriptor formulations. This behavior is consistent with the observed 
far-field boundary-condition violation. The failure to enforce the prescribed freestream state introduces an additional source of 
error into the flow solution, which is reflected in the slightly larger discrepancy in the predicted $C_p$ distribution. The 
descriptor and Laplace-shifted non-descriptor formulations, both of which satisfy the far-field condition, yield closely matching 
pressure distributions and more accurately reproduce the HDM solution.

Collectively, the NASA CRM results provide a complementary validation of the proposed framework. In this large-scale
three-dimensional problem, the numerical conditions identified as conducive to Type~3 snapping do not lead to actual snapshot
latching. Nevertheless, the underlying residual-scaling imbalance remains observable through the violation of the
parameter-dependent far-field condition by the unshifted non-descriptor formulation. The Laplace-governed shift restores this
boundary condition \textit{a priori} while retaining the non-descriptor formulation and its favorable near-field weighting. 
These results therefore demonstrate that the role of the state shift extends beyond suppressing observed snapshot snapping: it also 
provides a systematic mechanism for enforcing parameter-dependent far-field behavior while preserving the favorable near-field 
weighting of the non-descriptor formulation.


\section{Conclusions}
\label{sec:conclusions}

This work defined the snapping phenomenon in projection-based model reduction and established a unified framework for understanding 
its origins, manifestations, and mitigation. Three distinct mechanisms were identified. Type~1 snapping is intrinsic to the projection
process and arises when parameter-dependent variations become invisible to the reduced test space. Type~2 snapping is extrinsic and 
results from modeling choices such as localized reduction or hyperreduction. Type~3 snapping is numerical and originates from discrete
residual scaling, spatial resolution, and related numerical effects. This classification distinguishes mechanisms that are 
fundamentally tied to projection from those introduced by the formulation or discretization of the underlying computational model.

For linear parametric problems, the analysis established precise conditions under which a projection-based reduced-order model can 
stagnate at a training state despite a changing high-dimensional solution. In particular, necessary and sufficient conditions were 
derived for stagnation in the case of a parameter-independent operator, revealing that a parameter variation can become invisible 
when its forcing variation lies in the null space of the reduced test space. A parameter-dependent state shift was then introduced to 
restore sensitivity to such variations. The resulting formulation admits an offline--online decomposition in which the 
high-dimensional shift need not be explicitly constructed during the online stage. Under the stated assumptions, the shift eliminates 
the identified stagnation mechanism except at isolated parameter values associated with the interpolation structure.

The same state-shifting concept was subsequently extended to nonlinear finite-volume computational fluid dynamics, where Type~3 
snapping is associated with the spatial scaling of the residual. The analysis showed that the descriptor and non-descriptor 
formulations weight parameter-dependent boundary variations differently because of their respective cell-volume scalings. This 
explains why descriptor formulations can better preserve far-field parameter variations, whereas non-descriptor formulations provide 
stronger weighting of finely discretized near-wall regions. Rather than switching between these formulations according to the active 
parameterization, a Laplace-governed state shift was introduced within the non-descriptor formulation. The shift propagates 
parameter-dependent boundary transformations into the computational domain and enforces the corresponding boundary conditions 
\textit{a priori}, while retaining the near-field weighting of the non-descriptor residual.

The numerical results support these theoretical developments. The linear beam problem illustrates the intrinsic snapping mechanism 
and its mitigation through state shifting. The NACA airfoil problems demonstrate the numerical snapping associated with descriptor 
and non-descriptor residual scaling and show that the Laplace-governed shift can suppress the resulting snapshot latching while 
preserving the desired near-field behavior. The large-scale NASA CRM calculation further demonstrates the scalability of the approach 
to a three-dimensional viscous flow problem with nearly $10$ million nodes and more than $57$ million elements. Although no Type~3 
snapshot snapping was observed in this case, the expected far-field boundary-condition violation of the unshifted non-descriptor 
formulation was clearly observed. The Laplace-shifted non-descriptor formulation restored the prescribed far-field condition and 
produced a pressure distribution closely matching that obtained with the descriptor formulation.

The results establish that snapping is not a single failure mode of projection-based model order reduction but a collection of 
phenomena with distinct mathematical and numerical origins. The proposed state-shifting framework provides a common mechanism for 
restoring parameter sensitivity and enforcing parameter-dependent boundary conditions without requiring a change in the underlying 
reduced-order formulation. The framework therefore offers a systematic basis for diagnosing and mitigating snapping across linear 
and nonlinear parametric problems, while clarifying the respective roles of projection, model formulation, and discretization in 
the behavior of reduced-order models.

\section*{Acknowledgments}
The US authors acknowledge partial support from the Air Force Office of Scientific Research under Grants FA9550-22-1-0004 and 
FA9550-24-1-0327, and from the Office of Naval Research under Grant N00014-23-1-2877.





\printcredits

\bibliographystyle{model1-num-names}

\bibliography{cas-refs}

@book{wanner1996solving,
title={Solving ordinary differential equations II},
author={Wanner, Gerhard and Hairer, Ernst},
volume={375},
year={1996},
publisher={Springer Berlin Heidelberg New York}
}

@incollection{amsallem2014stability,
title={On the stability of reduced-order linearized computational fluid dynamics models based on POD and Galerkin projection: descriptor vs non-descriptor forms},
author={Amsallem, David and Farhat, Charbel},
booktitle={Reduced order methods for modeling and computational reduction},
pages={215--233},
year={2014},
publisher={Springer}
}

@inproceedings{washabaugh2014enforcing,
title={Enforcing Boundary Conditions for Discrete Nonlinear Aerodynamic Reduced-Order Models},
author={Washabaugh, Kyle M and Farhat, Charbel},
booktitle={32nd AIAA Applied Aerodynamics Conference},
pages={3169},
year={2014}
}

@phdthesis{washabaugh2016faster,
title={Faster fidelity for better design: A scalable model order reduction framework for steady aerodynamic design applications},
author={Washabaugh, Kyle Michael},
year={2016},
school={Stanford University}
}

@article{carlberg2011efficient,
title={Efficient non-linear model reduction via a least-squares Petrov--Galerkin projection and compressive tensor approximations},
author={Carlberg, Kevin and Bou-Mosleh, Charbel and Farhat, Charbel},
journal={International Journal for numerical methods in engineering},
volume={86},
number={2},
pages={155--181},
year={2011},
publisher={Wiley Online Library}
}

@article{carlberg2013gnat,
title={The GNAT method for nonlinear model reduction: effective implementation and application to computational fluid dynamics and turbulent flows},
author={Carlberg, Kevin and Farhat, Charbel and Cortial, Julien and Amsallem, David},
journal={Journal of Computational Physics},
volume={242},
pages={623--647},
year={2013},
publisher={Elsevier}
}

@article{de2026nonlinear,
title={Nonlinear projection-based model order reduction with machine learning regression for closure error modeling in the latent space},
author={De Parga, S Ares and Tezaur, Radek and Hern{\'a}ndez, Carlos G and Farhat, Charbel},
journal={Computer Methods in Applied Mechanics and Engineering},
volume={448},
pages={118443},
year={2026},
publisher={Elsevier}
}

@article{chaturantabut2010nonlinear,
title={Nonlinear model reduction via discrete empirical interpolation},
author={Chaturantabut, Saifon and Sorensen, Danny C},
journal={SIAM Journal on Scientific Computing},
volume={32},
number={5},
pages={2737--2764},
year={2010},
publisher={SIAM}
}

@article{farhat2014dimensional,
title={Dimensional reduction of nonlinear finite element dynamic models with finite rotations and energy-based mesh sampling and weighting for computational efficiency},
author={Farhat, Charbel and Avery, Philip and Chapman, Todd and Cortial, Julien},
journal={International Journal for Numerical Methods in Engineering},
volume={98},
number={9},
pages={625--662},
year={2014},
publisher={Wiley Online Library}
}

@article{farhat2015structure,
title={Structure-preserving, stability, and accuracy properties of the energy-conserving sampling and weighting method for the hyper reduction of nonlinear finite element dynamic models},
author={Farhat, Charbel and Chapman, Todd and Avery, Philip},
journal={International journal for numerical methods in engineering},
volume={102},
number={5},
pages={1077--1110},
year={2015},
publisher={Wiley Online Library}
}

@article{farhat2003application,
title={Application of a three-field nonlinear fluid--structure formulation to the prediction of the aeroelastic parameters of an F-16 fighter},
author={Farhat, Charbel and Geuzaine, Philippe and Brown, Gregory},
journal={Computers \& Fluids},
volume={32},
number={1},
pages={3--29},
year={2003},
publisher={Elsevier}
}

@misc{frg_aero_bitbucket,
author={Farhat Research Group},
title={AERO Suite Academic Code Base},
howpublished={\url{https://bitbucket.org/...}},
year={accessed 2026}
}

@inproceedings{vassberg2008development,
title={Development of a common research model for applied CFD validation studies},
author={Vassberg, John and Dehaan, Mark and Rivers, Melissa and Wahls, Richard},
booktitle={26th AIAA applied aerodynamics conference},
pages={6919},
year={2008}
}

@article{cosimo2016general,
title={General treatment of essential boundary conditions in reduced order models for non-linear problems},
author={Cosimo, Alejandro and Cardona, Alberto and Idelsohn, Sergio},
journal={Advanced Modeling and Simulation in Engineering Sciences},
volume={3},
number={1},
pages={7},
year={2016},
publisher={Springer}
}

@article{kalashnikova2014penalty,
title={A penalty/Lagrange multiplier formulation for enforcing boundary conditions in POD Galerkin reduced 
	order models}, 
author={Kalashnikova, Irina and van Bloemen Waanders, Bart and Arunajatesan, Srinivasan and Barone, Matthew F}, 
journal={Computer Methods in Applied Mechanics and Engineering}, 
volume={279},
pages={1--22},
year={2014},
publisher={Elsevier},
doi={10.1016/j.cma.2014.05.006}
}

@article{gunzburger2007reduced,
title={Reduced-order modeling of time-dependent PDEs with multiple parameters in the boundary data},
author={Gunzburger, Max D and Peterson, Janet S and Shadid, John N},
journal={Computer methods in applied mechanics and engineering},
volume={196},
number={4-6},
pages={1030--1047},
year={2007},
publisher={Elsevier}
}



\end{document}